\documentclass[10pt]{article}
\usepackage[a4paper, margin=1in]{geometry}
\usepackage{amsmath, amssymb, amsthm}
\usepackage{graphicx}
\usepackage{hyperref}
\usepackage{titling}
\usepackage{tocloft}
\usepackage{xcolor}
\usepackage{slashed}
\usepackage{geometry}
\usepackage{longtable}
\usepackage{caption}
\usepackage{tabularx}
\usepackage{mathtools}
\usepackage{algorithm}
\usepackage{algpseudocode}
\usepackage[all,cmtip]{xy} 

\newtheorem{theorem}{Theorem}
\newtheorem{lemma}{Lemma}
\newtheorem{corollary}{Corollary}
\newtheorem{definition}{Definition}
\newtheorem{conjecture}{Conjecture}
\newtheorem{proposition}{Proposition}

\newtheorem{remark}{Remark}

\newtheorem{assumption}{Assumption}

\newcommand{\TIME}{\mathsf{TIME}}
\newcommand{\SPACE}{\mathsf{SPACE}}
\DeclareMathOperator{\polylog}{polylog}
\DeclareMathOperator*{\argmax}{arg\,max}

\newcommand{\poly}{\mathsf{poly}}
\algrenewcommand\algorithmicrequire{\textbf{Input:}}
\algrenewcommand\algorithmicensure{\textbf{Output:}}

\algnewcommand\algorithmicto{\textbf{to}}
\algnewcommand\algorithmicforeach{\textbf{for each}}
\algrenewtext{For}[3]{\algorithmicfor\ #1 $\gets$ #2 \algorithmicto\ #3 \algorithmicdo}

\newcommand{\Front}{\mathrm{Front}}
\newcommand{\degmax}{\Delta_{\max}}
\newcommand{\pred}{\mathsf{pred}}

\hypersetup{
    colorlinks=true,
    linkcolor=blue,
    filecolor=magenta,
    urlcolor=blue,
    citecolor=blue,
    pdftitle={Universal Hirschberg},
    pdfpagemode=FullScreen,
}

\title{\Large \textbf{Universal Hirschberg for Width Bounded Dynamic Programs}\\[0.5in]}
\author{\Large Logan Nye, MD\\[0.3in]
Carnegie Mellon University School of Computer Science\\[0.1cm]
5000 Forbes Ave Pittsburgh, PA 15213 USA\\[0.5cm]
\texttt{lnye@andrew.cmu.edu}\\[0.2cm]
\small ORCID: \href{https://orcid.org/0009-0002-9136-045X}{0009-0002-9136-045X}}
\date{}

\begin{document}

\maketitle
\vspace{0.8in}

\begin{center}
    \Large \textbf{Abstract}
\end{center}
\vspace{0.3in}
Hirschberg’s algorithm (1975) reduces space complexity for the Longest Common Subsequence problem from $O(N^2)$ to $O(N)$ via recursive midpoint bisection on a grid DP. We show that the underlying idea generalizes to a broad class of dynamic programs with local dependencies on directed acyclic graphs (DP DAGs). Modeling a DP as deterministic time-evolution over a topologically ordered DAG with frontier width $\omega$ and bounded in-degree, and assuming a max-type semiring with deterministic tie-breaking, we prove that in a standard offline random-access model any such DP admits deterministic traceback in space $O(\omega \log T + \polylog T)$ cells over a fixed finite alphabet, where $T$ is the number of states. Our construction replaces backward DPs by forward-only recomputation and organizes the time order into a height-compressed recursion tree whose nodes expose small ``middle frontiers” across which every optimal path must pass. The framework yields near-optimal traceback bounds for asymmetric and banded sequence alignment, one-dimensional recurrences, and DP formulations on graphs of bounded pathwidth. We also show that an $\Omega(\omega)$ space term (in bits) is unavoidable in forward single-pass models and discuss conjectured $\sqrt{T}$-type barriers in streaming settings, supporting the view that space-efficient traceback is a structural property of width-bounded DP DAGs rather than a peculiarity of grid-based algorithms.

\newpage

\section{Introduction}

Dynamic programming (DP) is a central paradigm for discrete optimization and combinatorial problems. Given a DP formulation on a directed acyclic graph (DAG), we typically distinguish two tasks:

\begin{itemize}
  \item \emph{Value computation:} compute the optimal value $x_t$ at a designated sink $t$;
  \item \emph{Witness reconstruction:} reconstruct a specific optimal witness: e.g., a path from a source to $t$, an alignment, a labeling, or other structured solution.
\end{itemize}

For many DP formulations---on grids, banded lattices, and pathwidth- or treewidth-based decompositions---the value $x_t$ can be computed in space proportional to a structural width parameter, such as a frontier width $\omega$ or a pathwidth $\mathrm{pw}(G)$, by storing only a single frontier of DP states at a time. In contrast, witness reconstruction is usually implemented by storing the entire DP table or a large fraction of it, requiring $\Theta(T)$ space where $T$ is the number of states. This raises a question:

\begin{quote}
  \emph{For width-bounded DPs where the decision problem can be solved in $O(\omega)$ space, how much additional space is actually needed to recover a single optimal witness?}
\end{quote}

\paragraph{Hirschberg as a motivating example.}
A classical exception to the $\Theta(T)$ traceback pattern is Hirschberg's algorithm for LCS/edit distance~\cite{Hirschberg1975}. On an $m \times n$ grid DP ($T = \Theta(mn)$, frontier width $\omega = \min\{m,n\}$ under a natural row- or column-major topological order), naive traceback requires $O(mn)$ space. Hirschberg's divide-and-conquer algorithm recursively bisects along the longer dimension, uses forward and backward DPs to identify a midpoint on some optimal alignment, and recurses on the two subproblems. This reconstructs an optimal alignment in $O(m+n)$ space. The standard presentation is tightly coupled to the 2D grid geometry and the fact that the reverse DP has the same frontier width as the forward DP.

\paragraph{This work: a structural ``Universal Hirschberg'' theorem.}
We generalize Hirschberg's idea from grids to arbitrary \emph{time-ordered} DP DAGs with bounded frontier width, in an offline model with random access to the instance and inexpensive recomputation. Informally:

\begin{theorem}[Informal main result]
Let $G$ be a DP DAG on $T$ vertices equipped with a fixed topological order $\tau$, frontier width $\omega = \omega(G,\tau)$, bounded in-degree, and a max-type semiring recurrence with deterministic tie-breaking, so that each vertex has a unique witness predecessor. In the standard random-access multitape Turing-machine (or RAM) model, there is a deterministic algorithm that, given an encoding of $G$ and a designated sink $t$,
\begin{itemize}
  \item computes $x_t$ and outputs the canonical optimal source-to-$t$ witness path (as determined by the fixed tie-breaking rule),
  \item runs in time $\poly(T)$, and
  \item uses work-tape space
  \[
    \SPACE_{\mathrm{traceback}} \;=\; O\bigl(\omega \log T \;+\; \polylog T\bigr)
  \]
  measured in tape cells over a fixed finite alphabet. Equivalently, the algorithm stores the values of $O(\omega)$ DP states (each encoded in $O(\log T)$ bits) plus an additional $O(\polylog T)$ words of recursion metadata.
\end{itemize}
\end{theorem}

Thus, up to polylogarithmic factors, traceback can be done in essentially the \emph{same} space as value computation for any DP DAG whose frontier width is bounded under some topological order. No geometric symmetry, band structure, or small-width backward DP is required; the only structural parameter that enters the space bound is the frontier width $\omega$.

\paragraph{Techniques and scope.}
Our algorithm organizes the fixed topological order into a balanced binary recursion tree over intervals of ``time''. Each node induces a subproblem whose internal vertices lie in a time window $I$, with small interfaces to the rest of the DAG. At each level, we:

\begin{itemize}
  \item identify a \emph{middle frontier} $F^{\mathrm{mid}}_I$ across which every optimal path relevant to the subproblem must pass;
  \item compute prefix values from the left boundary to each $v \in F^{\mathrm{mid}}_I$ by a single forward DP on $I$;
  \item compute suffix values from each $v \in F^{\mathrm{mid}}_I$ to the sink $t$ via forward-only recomputation on subintervals (no backward DP on the reversed graph);
  \item select a canonical midpoint vertex $v^\star$ using the max-type semiring and deterministic tie-breaking; and
  \item recurse on the left and right subintervals, carrying only boundary values and a small amount of recursion metadata.
\end{itemize}

The working frontier buffer always stores at most $\omega$ DP values at any time, and the recursion tree has depth $O(\log T)$; recursion metadata fits in $O(\polylog T)$ space. We work throughout with finite DAGs of bounded in-degree, max-type semirings with deterministic witnesses, and deterministic algorithms in a random-access model with running time polynomial in $T$.

\paragraph{Applications and limitations.}
Instantiating the framework yields:

\begin{itemize}
  \item for asymmetric sequence alignment on an $m \times n$ grid with $m \le n$, traceback in space $O(m \log(mn) + \polylog(mn))$;
  \item for one-dimensional recurrences with constant in-degree, traceback in space $O(\polylog T)$;
  \item for banded alignment with bandwidth $B$, traceback in space $O(B \log(BN) + \polylog(BN))$; and
  \item for DPs on DAGs of pathwidth $\mathrm{pw}(G)$, traceback in space $O\bigl(\mathrm{pw}(G) \log |V| + \polylog |V|\bigr)$.
\end{itemize}

We also show that an $\Omega(\omega)$ term (in bits) is unavoidable in any forward, single-pass model where the input is streamed once in topological order, and we discuss conjectured $\sqrt{T}$-type barriers in genuinely streaming settings, in contrast to the offline random-access model we analyze.

The rest of the paper develops the formal model and assumptions (Section~\ref{sec:model}), recasts Hirschberg's algorithm in structural terms (Section~\ref{sec:warmup}), presents the height-compressed recursion and universal traceback algorithm (Sections~\ref{sec:height-compression}--\ref{sec:universal}), derives applications (Section~\ref{sec:applications}), discusses lower bounds and model variants (Section~\ref{sec:lower-bounds}), and concludes with related work and open problems (Sections~\ref{sec:related}--\ref{sec:open}).

\section{Model and Assumptions}
\label{sec:model}

We now formalize the setting of dynamic programming as a graph-theoretic problem on a structured directed acyclic graph. We define the parameters governing the space complexity of evaluation and the specific class of optimization problems---max-type recurrences with unique witnesses---to which our results apply.

\subsection{Dynamic programming as a local DAG}

We abstract the dependency structure of a dynamic program into a directed acyclic graph where vertices represent states and edges represent dependencies in the recurrence.

\begin{definition}[Computational DP DAG]
A \emph{computational DP DAG} is a tuple
\[
  G \;=\; (V, E, S_{\mathrm{src}}, T_{\mathrm{sink}}, \tau)
\]
where:
\begin{itemize}
  \item $V$ is a finite set of $T = |V|$ vertices, which we refer to as \emph{states},
  \item $E \subseteq V \times V$ is a set of directed edges representing computational dependencies,
  \item $S_{\mathrm{src}} \subseteq V$ is a non-empty set of \emph{source vertices} (states with in-degree $0$),
  \item $T_{\mathrm{sink}} \subseteq V$ is a non-empty set of \emph{sink vertices},
  \item $\tau : V \to \{1, \dots, T\}$ is a bijection defining a \emph{topological order}: for every edge $(u, v) \in E$, we have $\tau(u) < \tau(v)$.
\end{itemize}
\end{definition}

For notational convenience, we identify the vertices with their topological indices; henceforth, we assume $V = \{1, \dots, T\}$ and $\tau(v) = v$, so an edge $(u, v)$ implies $u < v$. All references to cuts, frontiers, and frontier width are with respect to this fixed topological order. We are free to choose $\tau$ when formulating the DP; in particular, in applications where the underlying graph has bounded pathwidth, we will later choose $\tau$ derived from a path decomposition.

Standard dynamic programming formulations typically involve recurrences with a constant number of terms. We formalize this via a bounded in-degree assumption.

\begin{assumption}[Locality]
\label{ass:locality}
There exists a universal constant $\degmax \in \mathbb{N}$ such that for all $v \in V$, the in-degree is bounded:
\[
  \deg^{-}(v) \;=\; |\{u : (u, v) \in E\}| \;\le\; \degmax.
\]
\end{assumption}

This assumption implies that the evaluation of a single state, given the values of its predecessors, requires $O(1)$ primitive semiring operations.

\subsection{Time order and frontier width}

The space complexity of evaluating a DP is governed not by the total size of the state space, but by the maximum number of values that must be persisted across any cut in the topological order. We formalize this memory footprint as the \emph{frontier width}. The following definition is specific to the fixed topological order $\tau$; changing $\tau$ may change the frontier width.

\begin{definition}[Frontier and frontier width]
For any time step $\ell \in \{0, \dots, T\}$, the \emph{frontier at time $\ell$}, denoted $\Front(\ell)$, is the set of vertices processed at or before time $\ell$ that have outgoing edges to vertices processed strictly after $\ell$:
\[
  \Front(\ell) \;:=\; \bigl\{ v \in V : v \le \ell \text{ and } \exists w > \ell \text{ such that } (v, w) \in E \bigr\}.
\]
The \emph{frontier width} of $G$ with respect to $\tau$, denoted $\omega(G,\tau)$, is the maximum frontier size:
\[
  \omega(G,\tau) \;:=\; \max_{0 \le \ell \le T} |\Front(\ell)|,
\]
and we write $\omega$ when $G$ and $\tau$ are clear from context.
\end{definition}

The frontier $\Front(\ell)$ captures precisely the set of DP state values that a forward-pass evaluation algorithm must retain in memory after step $\ell$ in order to continue computing all future states $w > \ell$. In geometric settings such as the standard Needleman--Wunsch algorithm on an $m \times n$ grid (with $m \le n$), a natural topological order proceeds column-by-column or row-by-row; the frontier then corresponds to a single column or row, and $\omega(G,\tau) = \Theta(\min\{m,n\})$ for that choice of $\tau$.

\subsection{Semiring and witness structure}

We restrict our attention to optimization problems over ordered structures where one seeks not just the optimal value, but a canonical certificate (witness) of that optimality.

\begin{assumption}[Max-type semiring with deterministic witnesses]
\label{ass:max}
We consider computations over a structure $(\mathcal{S}, \oplus, \otimes, \preceq)$ satisfying:
\begin{enumerate}
  \item $(\mathcal{S}, \preceq)$ is a totally ordered set.
  \item The ``addition'' operator $\oplus$ is the supremum with respect to $\preceq$ over finite subsets of $\mathcal{S}$, and in particular for any $a,b \in \mathcal{S}$ we have $a \oplus b = \max\nolimits_{\preceq}\{a, b\}$. There is a distinguished bottom element $\mathbf{0} \in \mathcal{S}$ such that $\mathbf{0} \oplus a = a$ for all $a$.
  \item The ``multiplication'' operator $\otimes$ is associative and distributes over $\oplus$.
  \item $\otimes$ is monotone with respect to $\preceq$: if $a \preceq b$, then $a \otimes c \preceq b \otimes c$ and $c \otimes a \preceq c \otimes b$ for all $c \in \mathcal{S}$.
\end{enumerate}

For each vertex $v \in V$, we associate a value $x_v \in \mathcal{S}$ defined by the recurrence
\[
  x_v \;=\;
  \begin{cases}
      a_v & \text{if } v \in S_{\mathrm{src}}, \\[3pt]
      \displaystyle \bigoplus\limits_{(u, v) \in E} \bigl(x_u \otimes w_{u,v}\bigr) & \text{otherwise},
  \end{cases}
\]
where $a_v \in \mathcal{S}$ is an initialization value for sources and $w_{u,v} \in \mathcal{S}$ is the weight of edge $(u, v)$.

We assume that every semiring element (weight or DP value) admits a representation using $O(\log T)$ tape cells over a fixed finite alphabet. We refer to this unit of storage ($O(\log T)$ bits) as a \emph{machine word}. We assume that the primitive operations (comparison, $\oplus$, $\otimes$) as well as index arithmetic on $\{1, \dots, T\}$ can be performed in time and space polynomial in $\log T$.

Crucially, we assume a fixed, deterministic tie-breaking rule on edges. Fix once and for all a total order $\prec$ on the edge set $E$. For any vertex $v \notin S_{\mathrm{src}}$, let $P(v) = \{u : (u, v) \in E\}$ be the set of predecessors. We define the unique \emph{witness predecessor} $\pred(v)$ as follows:
\[
  \pred(v) \;=\; u^\star \quad\text{where } u^\star \in P(v),\; x_{u^\star} \otimes w_{u^\star,v} = x_v,
\]
and $(u^\star,v)$ is $\prec$-minimal among all edges $(u,v)$ with $u \in P(v)$ satisfying $x_u \otimes w_{u,v} = x_v$. For $v \in S_{\mathrm{src}}$, we set $\pred(v) = \bot$.
\end{assumption}

The tie-breaking rule ensures that, once the DP values $\{x_v\}_{v \in V}$ are fixed, the predecessor map $\pred : V \to V \cup \{\bot\}$ is uniquely determined.

\begin{lemma}[Consistency of local re-evaluation]
\label{lem:local-consistency}
Let $U \subseteq V$ be any subset of vertices, and let $B \subseteq V \setminus U$ be a boundary set such that:
\begin{itemize}
  \item for every edge $(u,v) \in E$ with $v \in U$ and $u \notin U$, we have $u \in B$, and
  \item the subgraph induced by $U \cup B$ is acyclic and respects the restriction of the topological order $\tau$.
\end{itemize}
Suppose we are given the true global DP values $\{x_u\}_{u \in B}$ on $B$. If we recompute the DP on the induced subgraph on $U \cup B$ using these boundary values and the same recurrence and tie-breaking rule, then for every $v \in U$ the recomputed value $x_v'$ and predecessor $\pred'(v)$ coincide with the global value $x_v$ and predecessor $\pred(v)$.
\end{lemma}

\begin{proof}
Order the vertices of $U \cup B$ according to the restriction of $\tau$. We prove the claim by induction on this order, restricted to $U$. 

\emph{Base case:} Let $v$ be the first vertex in $U$ according to $\tau$. If $v \in S_{\mathrm{src}}$, its value depends only on initialization $a_v$, so $x_v' = x_v$ trivially. If $v \notin S_{\mathrm{src}}$, all its predecessors $u$ lie in $B$ (since any predecessor in $U$ would precede $v$ in $\tau$, contradicting that $v$ is first). The values $x_u$ for $u \in B$ are equal to the global values by assumption. Thus, $x_v'$ is computed from the same multiset $\{x_u \otimes w_{u,v}\}$ as the global value, and tie-breaking uses the same $\prec$, so $x_v' = x_v$ and $\pred'(v) = \pred(v)$.

\emph{Inductive step:} For a general $v \in U$, predecessors lie either in $B$ or in $U$ (earlier in the order). By the induction hypothesis, predecessors in $U$ have correct values $x_u' = x_u$. Predecessors in $B$ have correct values by assumption. The calculation for $v$ is therefore identical to the global calculation.
\end{proof}

This determinism and consistency are essential: they guarantee that for any optimal sink value $x_t$, there exists a unique canonical trajectory through the DAG obtained by iteratively following $\pred(\cdot)$, and that re-evaluating the recurrence in any local window with correct boundary values consistently yields the same predecessor choices. This enables space-efficient reconstruction without storing the global pointer graph.

\begin{definition}[Traceback problem]
Given a computational DP DAG $G$ and a target sink $t \in T_{\mathrm{sink}}$, the \emph{traceback problem} asks to construct the unique sequence of vertices $P = (v_0, v_1, \dots, v_k)$ such that:
\begin{enumerate}
  \item $v_k = t$,
  \item $v_0 \in S_{\mathrm{src}}$, and
  \item for all $1 \le i \le k$, $v_{i-1} = \pred(v_i)$.
\end{enumerate}
We refer to $P$ as the canonical optimal witness path from a source to $t$.
\end{definition}

\subsection{Machine model and space measure}

We analyze space complexity in the standard deterministic multitape Turing machine model, which is polynomially equivalent (for space) to the Random Access Machine (RAM) model.

\begin{itemize}
  \item \textbf{Input:} The graph $G$ (adjacency lists, edge weights) is provided on a read-only input tape. The input is \emph{random-access} in the sense that the machine can move its input head to any vertex or edge description in time polynomial in $\log T$.
  \item \textbf{Output:} The witness path $P$ is written to a write-only output tape. This tape does not contribute to the space bound.
  \item \textbf{Work tape:} The machine has one or more read/write work tapes. Space is measured as the number of cells used on these work tapes over a fixed finite alphabet~$\Gamma$.
  \item \textbf{Computational primitives:} We assume that operations in the semiring (comparison, $\oplus$, $\otimes$) and index arithmetic on $\{1, \dots, T\}$ can be performed in time and space polynomial in $\log T$.
\end{itemize}

Throughout the paper, all space bounds refer to cells on the work tapes; the read-only input tape and the write-only output tape do not contribute to the space measure. When we informally say that an algorithm uses $O(\omega)$ space (in words), we mean that it stores the values of at most $O(\omega)$ DP states at any time, which corresponds to $O(\omega \log T)$ work-tape cells under our encoding assumptions. We restrict attention to algorithms that halt in time $\poly(T)$. Our primary objective is to minimize the work-tape space required to produce the output sequence $P$, as a function of the structural parameters $T$ and $\omega$.

\section{Warm-Up: Hirschberg as a Special Case}
\label{sec:warmup}

Before addressing the general graph-theoretic setting, we revisit Hirschberg's classic algorithm for the Longest Common Subsequence (LCS) problem. We view it not as a special-purpose geometric trick, but as an instance of a more general pattern: a recursive decomposition of a topological order together with meet-in-the-middle reasoning about optimal paths. This perspective will directly motivate our generalization to arbitrary DAGs with bounded frontier width.

\subsection{Classical grid DP}

Let $A = a_1 \dots a_m$ and $B = b_1 \dots b_n$ be sequences over an alphabet $\Sigma$. The standard dynamic programming formulation for LCS (or, with modified weights, for edit distance or global alignment) constructs a grid graph $G_{m,n}$ with vertex set
\[
  V \;=\; \{ (i,j) : 0 \le i \le m,\; 0 \le j \le n \}.
\]
Directed edges connect $(i,j)$ from $(i-1,j)$, $(i,j-1)$, and $(i-1,j-1)$, with edge weights determined by the scoring scheme (match/mismatch, insertion, deletion). The DP value $x_{i,j}$ at $(i,j)$ is defined by a max-type recurrence, e.g.\ for LCS:
\[
  x_{i,j} \;=\;
  \begin{cases}
    0 & \text{if } i=0 \text{or } j=0,\\[2pt]
    \max\{ x_{i-1,j},\, x_{i,j-1},\, x_{i-1,j-1} + \mathbf{1}[a_i = b_j]\} & \text{otherwise},
  \end{cases}
\]
with deterministic tie-breaking to select a unique predecessor in case of equality.

In the language of Section~\ref{sec:model}:
\begin{itemize}
  \item the graph size is $T = |V| = (m+1)(n+1)$;
  \item a natural topological order $\tau$ is lexicographical (row-major or column-major), e.g.\ $\tau(i,j) = i(n+1) + j$;
  \item for either of these topological orders, any cut in the order intersects $\Theta(\min\{m,n\})$ DP states, so the frontier width is $\omega(G_{m,n},\tau) = \Theta(\min\{m,n\})$;
  \item the recurrence is over a max-type semiring (e.g.\ $(\mathbb{Z}, \max, +)$) with constant in-degree $\degmax = 3$, and deterministic tie-breaking yields a unique predecessor at each vertex.
\end{itemize}

A standard forward implementation maintains a rolling buffer of one row (or one column), evaluating all $T$ states in $O(T)$ time and $O(\omega)$ space (in words) and computing the optimal value $x_{m,n}$. However, reconstructing an optimal alignment naively requires storing a back-pointer at each cell, resulting in $\Theta(mn)$ words for traceback.

\subsection{Hirschberg's algorithm: a structural view}

Hirschberg's algorithm circumvents the quadratic traceback space requirement by exploiting the product structure of the grid and the reversibility of the recurrence. We rephrase its logic as a recursive identification of optimal-path crossings through a balanced separator in the topological order. This makes the structural pattern more visible and prepares the ground for generalization.

For simplicity, assume $m \le n$ and use a column-major order so that the frontier width is $\omega = \Theta(m)$.

The algorithm proceeds as follows:
\begin{enumerate}
  \item \textbf{Midpoint selection.}
  Fix a column index $k = \lfloor n/2 \rfloor$. This defines a vertical cut through the grid, separating the problem into a \emph{left} half (columns $0,\dots,k$) and a \emph{right} half (columns $k,\dots,n$). In terms of the topological order, this is a balanced cut of the time axis.

  \item \textbf{Forward valuation.}
  Compute, using space $O(m)$ words, the optimal value
  \[
    f(i) \;=\; \text{optimal score of a path from $(0,0)$ to $(i,k)$}
  \]
  for all $0 \le i \le m$, via a standard forward DP restricted to columns $0,\dots,k$.

  \item \textbf{Backward valuation.}
  Compute, again using $O(m)$ words, the optimal value
  \[
    g(i) \;=\; \text{optimal score of a path from $(i,k)$ to $(m,n)$}
  \]
  for all $0 \le i \le m$. This step relies on the symmetry of the grid: the backward problem is isomorphic to a forward DP on the reversed strings $A^\mathrm{rev}$, $B^\mathrm{rev}$ over a grid of the same width $m$, with the same frontier width $\Theta(m)$ under the reversed topological order.

  \item \textbf{Crossing identification.}
  By the principle of optimality, any optimal path from $(0,0)$ to $(m,n)$ must cross the column $k$. Thus there exists some row $i^\star$ such that the vertex $(i^\star,k)$ lies on an optimal path. Moreover,
  \[
    i^\star \;=\; \argmax_{0 \le i \le m} \bigl( f(i) + g(i) \bigr),
  \]
  and $(i^\star,k)$ is a canonical midpoint vertex on an optimal witness path when the same deterministic tie-breaking rule is used in both the forward and backward DPs.

  \item \textbf{Recursion.}
  The problem decomposes into two independent subproblems:
  \begin{itemize}
    \item an LCS (or alignment) problem from $(0,0)$ to $(i^\star,k)$ on the left subgrid, and
    \item an LCS problem from $(i^\star,k)$ to $(m,n)$ on the right subgrid.
  \end{itemize}
  The algorithm recurses on these two subproblems, accumulating the witness path segments.
\end{enumerate}

\paragraph{Space complexity.}
Let $S(m,n)$ denote the work-tape space required (in machine words) to reconstruct an alignment via Hirschberg, in the machine model of Section~\ref{sec:model} where the input strings $A$ and $B$ reside on a read-only input tape. At each recursive call, the algorithm maintains two line buffers of length $\Theta(m)$ for the forward and backward passes, plus a constant amount of metadata (including the current column index and recursion state) on the recursion stack. The recursion depth is $O(\log n)$, and the additional per-level metadata is $O(1)$ words. Thus we obtain the recurrence
\[
  S(m,n) \;\le\; \max\{ S(m, \lfloor n/2 \rfloor),\, O(m) \} \;+\; O(1),
\]
which solves to
\[
  S(m,n) \;=\; O(m + \log n).
\]
In particular, the traceback workspace is linear in the shorter dimension $m$, plus a logarithmic term for the recursion depth, representing a substantial improvement over the naive $\Theta(mn)$ pointer-storage scheme in regimes where $m \ll n$. Translating to tape cells using our encoding assumptions, this corresponds to $O\bigl((m + \log n)\log(mn)\bigr)$ cells.

\subsection{Key structural ideas}

Abstracting away the grid geometry reveals two structural components driving Hirschberg's efficiency:

\begin{itemize}
  \item \textbf{Topological bisection.}
  The column $k$ acts as a balanced separator in the topological order induced by the grid. Splitting the problem at $k$ ensures the recursion depth is $O(\log n)$, i.e., logarithmic in one dimension of the topological order.

  \item \textbf{Meet-in-the-middle evaluation.}
  The set of candidate ``midpoint'' vertices consists of those on the cut column $k$ that lie on some path from source to sink. By combining forward values $f(i)$ and backward values $g(i)$ at these midpoints, the algorithm locates a vertex which \emph{must} lie on an optimal path. When the same deterministic tie-breaking is used in all DP evaluations, the argmax of $f(i) + g(i)$ picks out a canonical midpoint along the canonical optimal path.
\end{itemize}

Our generalization addresses the following question:
\begin{quote}
  \emph{Can we replicate this divide-and-conquer strategy on an arbitrary DP DAG where no geometric symmetry exists and where ``backward'' evaluation may be computationally expensive or wide?}
\end{quote}

We answer in the affirmative. We will demonstrate that any DAG with small frontier width admits a recursive decomposition into balanced intervals of the topological order, and that one can perform meet-in-the-middle traceback using \emph{only forward} DP evaluations. The key idea is to recompute suffix costs via local forward DPs, trading time for space, and to organize the recursion via a height-compressed decomposition of the time line.

\section{Height-Compressed Recursion for DP DAGs}
\label{sec:height-compression}

We now construct the structural framework for our algorithm: a hierarchical decomposition of the time-ordered DP DAG together with a notion of local interfaces. This construction mirrors the height-compressed computation trees employed in time--space tradeoff results for Turing machines, but here we exploit random access to the input and are free to recompute local DPs as often as needed while keeping the working frontier bounded by the frontier width.

\subsection{Recursive decomposition of the time order}

Recall that $\tau : V \to \{1,\dots,T\}$ defines a fixed topological ordering of the vertices. For the remainder of the exposition, we identify each vertex with its topological index, setting $V = \{1,\dots,T\}$ and $\tau(v) = v$. Edges $(u,v) \in E$ thus satisfy $u < v$.

We partition the computation not by geometry, but by contiguous intervals of topological time.

\begin{definition}[Interval and midpoint split]
Let $I = [i,j]$ be an interval of indices with $1 \le i \le j \le T$. We define the \emph{midpoint} of $I$ as
\[
  m(I) \;:=\; \bigl\lfloor (i+j)/2 \bigr\rfloor.
\]
The interval $I$ is partitioned into two disjoint sub-intervals: the \emph{left child} $I_L = [i, m(I)]$ and the \emph{right child} $I_R = [m(I)+1, j]$.
\end{definition}

This bisection scheme induces a canonical binary tree over the time axis.

\begin{definition}[Recursion tree]
The \emph{recursion tree} $\mathcal{R}$ for the time domain $[1,T]$ is the rooted binary tree defined inductively:
\begin{itemize}
  \item the root of $\mathcal{R}$ is the interval $[1,T]$;
  \item every internal node $I = [i,j]$ (with $i<j$) has exactly two children: $I_L$ and $I_R$;
  \item the leaves of $\mathcal{R}$ are singleton intervals $[k,k]$ with $k \in \{1,\dots,T\}$.
\end{itemize}
In our algorithm we will stop the recursion earlier, at a base-case size $B_{\mathrm{base}} \ge 1$, so leaves will correspond to intervals of length at most $B_{\mathrm{base}}$.
\end{definition}

\begin{lemma}[Logarithmic depth]
\label{lem:depth}
The depth of the recursion tree $\mathcal{R}$ is $O(\log T)$. Moreover, for any node $I$ at depth $d$, the length of the interval satisfies $|I| \le \lceil T / 2^d \rceil$.
\end{lemma}

\begin{proof}
Let $L(d)$ denote the maximum length of an interval at depth $d$. By construction, $L(0) = T$, and for any depth $d$, each child interval has length at most $\lceil L(d)/2 \rceil$. Thus $L(d+1) \le \lceil L(d)/2 \rceil$, which implies by induction that
\[
  L(d) \;\le\; T/2^d + 1.
\]
In particular, $L(d) \le 1$ once $d \ge \log_2 T + 1$, so the depth is $O(\log T)$.
\end{proof}

This decomposition provides a ``coordinate system'' for our recursive traceback: each node $I$ of $\mathcal{R}$ corresponds to a subproblem in which we care about paths whose internal vertices lie in the time window $I$.

\subsection{Interval interfaces and local DPs}

To evaluate the DP on a sub-interval $I$, we must formally define the input and output boundaries of the induced subgraph.

\begin{definition}[Restricted DAG and interfaces]
Let $I = [i,j]$ be an interval in $\mathcal{R}$.
\begin{itemize}
  \item The \emph{restricted vertex set} is $V_I = \{ v \in V : v \in I \}$.
  \item The \emph{restricted edge set} is $E_I = \{ (u,v) \in E : u,v \in V_I \}$.
  \item The \emph{left interface} (or input boundary) is the set of vertices within $I$ that depend on computations prior to $i$:
  \[
    F_I^{\mathrm{in}} \;:=\; \{ v \in V_I : \exists u < i \text{ with } (u,v) \in E \}.
  \]
  \item The \emph{right interface} (or output boundary) is the set of vertices within $I$ that are needed by computations after $j$:
  \[
    F_I^{\mathrm{out}} \;:=\; \{ v \in V_I : \exists w > j \text{ with } (v,w) \in E \}.
  \]
\end{itemize}
\end{definition}

It is important to distinguish the input interface (receivers) from the actual boundary (senders). The space required to restart the DP in $I$ is determined by the set of vertices \emph{outside} $I$ that have edges into $I$. We define the \emph{left boundary set} as:
\[
  B_I \;:=\; \bigl\{ u < i : \exists v \in V_I \text{ with } (u,v) \in E \bigr\}.
\]
By construction, $B_I \subseteq \Front(i-1)$ and $F_I^{\mathrm{out}} \subseteq \Front(j)$, where the frontiers are defined with respect to the fixed topological order $\tau$. In particular,
\[
  |B_I| \;\le\; \omega
  \qquad\text{and}\qquad
  |F_I^{\mathrm{out}}| \;\le\; \omega.
\]
Thus, every subproblem induced by an interval requires inputs from a set $B_I$ of size at most $\omega$, and produces outputs relevant for the future on a set of size at most $\omega$.

We next formalize the fact that we can evaluate the DP within $I$ in a space bounded by $\omega$, given boundary values on a small set of vertices before $i$.

\begin{lemma}[Local DP evaluation]
\label{lem:local-dp}
Let $I = [i,j]$ be an interval, and let $U \subseteq V_I$ be any target set. Suppose we are given the true global DP values $\{x_u\}_{u \in B_I}$ on the boundary set $B_I$. Then the values $\{x_v\}_{v \in U}$ determined by the global recurrence (Assumption~\ref{ass:max}) can be recomputed using additional workspace $O(\omega)$ (in words).
\end{lemma}

\begin{proof}
We describe a forward-pass algorithm $\mathcal{A}_I$ that recomputes the DP values for vertices in $I$, assuming the values on $B_I$ are available.

Maintain a working buffer $M$ (the \emph{active set}) intended to store the current active values within $I$ together with the fixed boundary values. We will ensure $|M| \le \omega$ at all times.

\begin{enumerate}
  \item Initialize $M$ to contain the pairs $(u,x_u)$ for all $u \in B_I$.
  \item For $k$ increasing from $i$ to $j$:
  \begin{itemize}
    \item Compute $x_k$ from its predecessors using the recurrence
    \[
      x_k \;=\; \bigoplus_{(u,k) \in E} (x_u \otimes w_{u,k}),
    \]
    where each predecessor $u$ is either in $B_I$ (if $u < i$) or has already been processed in $\{i,\dots,k-1\}$. For $u \in B_I$, $x_u$ is part of the boundary, and for $u \in [i,k-1]$ we maintain $x_u$ in $M$ by induction.
    \item Insert $(k,x_k)$ into $M$.
    \item Remove from $M$ any vertex $u$ that will never be used again as a predecessor of any vertex in $I$ with index $>k$. Formally, we remove $u$ if there is no $w \in I$ with $w>k$ and $(u,w) \in E$. To test this, for each $u \in M$ we scan the adjacency list of $u$ on the read-only input tape and check whether there exists such a successor in $I$; if not, we delete $u$ from $M$.
  \end{itemize}
  \item Whenever $k \in U$, record (or output) the value $x_k$.
\end{enumerate}

At any time $k$, the set of vertices $u$ that have a successor in $I$ with index $>k$ is contained in $\Front(k) \cap (B_I \cup V_I)$, whose size is at most $\omega$ by definition of frontier width. By construction $M$ never contains any vertex outside this set, so $|M| \le \omega$ throughout the computation. The algorithm uses only the working buffer $M$ plus a constant amount of scratch space for scanning adjacency lists and indices, so the additional work-tape space beyond the representation of the boundary values is $O(\omega)$.

Finally, by Lemma~\ref{lem:local-consistency}, recomputing the DP in this way using the true boundary values yields the same $x_v$ as in the global DP, for all $v \in U$.
\end{proof}

This lemma guarantees that, modulo $O(\omega)$ workspace, we can compute:
\begin{itemize}
  \item \textbf{Forward costs.} Given values on a left boundary $B_I \subseteq \Front(i-1)$, we can compute the optimal value to reach any $v \in F_I^{\mathrm{out}}$ via paths whose internal vertices lie in $I$.
  \item \textbf{Suffix costs.} Given a vertex $v \in I$ and a sink $s \in I$ with $v \le s$, we can similarly run a forward DP on the induced subgraph of $I$ starting from $v$ to compute the optimal value of a path from $v$ to $s$ whose internal vertices lie in $I$. As with the prefix computation, the active set remains a subset of the global frontier, ensuring the $O(\omega)$ space bound.
\end{itemize}

In the next section, we will use these local evaluations as building blocks in a recursive, meet-in-the-middle traceback algorithm structured by the recursion tree $\mathcal{R}$. The global-to-local consistency of Lemma~\ref{lem:local-consistency} will ensure that whenever we seed a local DP with true boundary values, the recomputed predecessors agree with the unique global predecessor map on the relevant vertices.

\section{Universal Hirschberg Algorithm}
\label{sec:universal}

We now present the core algorithm and its complexity analysis. The construction generalizes Hirschberg's divide-and-conquer strategy from grids to arbitrary DP DAGs, achieving space efficiency by trading re-computation for memory. Structurally, the algorithm is a meet-in-the-middle traceback scheme organized along the recursion tree $\mathcal{R}$ from Section~\ref{sec:height-compression}, operating in the offline random-access model of Section~\ref{sec:model}.

\subsection{Problem restatement and invariant}

Fix a DP DAG $G = (V,E,S_{\mathrm{src}},T_{\mathrm{sink}},\tau)$ satisfying Assumptions~\ref{ass:locality} and~\ref{ass:max}, and let $t \in T_{\mathrm{sink}}$ be a designated sink. Let $x_t$ denote its optimal value under the global recurrence.

The traceback algorithm maintains recursive subproblems of the following form.

\begin{quote}
  \textbf{Traceback subproblem} $(I, B_{\mathrm{left}}, s, \mathrm{val})$:
  \begin{itemize}
    \item $I = [i,j]$ is an interval in the recursion tree $\mathcal{R}$.
    \item $B_{\mathrm{left}}$ is a boundary set of vertices with indices $< i$ or in $S_{\mathrm{src}}$ (possibly empty at internal levels) with known values $\{x_u\}_{u \in B_{\mathrm{left}}}$.
    \item $s \in V_I$ is a target vertex with $i \le s \le j$.
    \item $\mathrm{val} \in \mathcal{S}$ is the optimal value of a path from \emph{some} $u \in B_{\mathrm{left}}$ to $s$ whose internal vertices lie in $I$.
  \end{itemize}
\end{quote}

Intuitively, $(I,B_{\mathrm{left}},s,\mathrm{val})$ represents a subproblem in which we have ``compressed'' everything before time $i$ into boundary vertices and wish to reconstruct the portion of the canonical witness path whose internal vertices lie between $i$ and $j$.

The algorithm maintains the following invariant.

\begin{lemma}[Subproblem invariance]
\label{lem:subproblem-invariant}
In the initial call and at every recursive call of the traceback procedure, the subproblem $(I,B_{\mathrm{left}},s,\mathrm{val})$ satisfies:
\begin{enumerate}
  \item $B_{\mathrm{left}} \subseteq \{1,\dots,i-1\} \cup S_{\mathrm{src}}$ and $s \in I$,
  \item there exists at least one optimal witness path $P$ from $B_{\mathrm{left}}$ to $s$ whose internal vertices all lie in $I$ and whose semiring value is $\mathrm{val}$.
\end{enumerate}
In particular, for the initial call we have $\mathrm{val} = x_t$.
\end{lemma}

\begin{proof}
For the initial call, we set $I = [1,T]$, $B_{\mathrm{left}} = S_{\mathrm{src}}$, $s = t$, and $\mathrm{val} = x_t$. By definition of $x_t$, there exists an optimal path from $S_{\mathrm{src}}$ to $t$ with value $x_t$, and all its internal vertices lie in $I = [1,T]$. Thus items (1) and (2) hold. Note that for $I=[1,T]$, the condition $B_{\mathrm{left}} \subseteq \{1,\dots,0\} \cup S_{\mathrm{src}}$ simplifies to $B_{\mathrm{left}} \subseteq S_{\mathrm{src}}$, which holds.

Assume the lemma holds for a parent subproblem $(I,B_{\mathrm{left}},s,\mathrm{val})$ with $I = [i,j]$. Let $m = m(I)$.

\begin{itemize}
  \item If $s \le m$, then any path from $B_{\mathrm{left}}$ to $s$ contained in $I$ must lie entirely in $[i,m] = I_L$, as indices strictly increase. Since an optimal path $P$ exists for the parent, that same $P$ is valid for the child subproblem $(I_L,B_{\mathrm{left}},s,\mathrm{val})$. Thus the invariant holds.

  \item If $s > m$, let $P$ be an optimal witness path guaranteed by the parent invariant. Lemma~\ref{lem:crossing} (proved below) shows that $P$ intersects the middle frontier $F^{\mathrm{mid}}_I$ at some vertex $v^\circ$. Proposition~\ref{prop:decomposition} (proved later) establishes that
  \[
      \mathrm{val} \;=\; \max_{v \in F^{\mathrm{mid}}_I} (f(v) \otimes g(v)).
  \]
  The algorithm selects a canonical maximizer $v^\star \in F^{\mathrm{mid}}_I$ deterministically. Since $v^\star$ achieves the maximum, there exists an optimal path $P^\star$ passing through $v^\star$ with value $\mathrm{val}$.
  
  Split $P^\star$ at $v^\star$ into a prefix $P_{\mathrm{pref}}$ (from $B_{\mathrm{left}}$ to $v^\star$) and a suffix $P_{\mathrm{suff}}$ (from $v^\star$ to $s$). The value of $P_{\mathrm{pref}}$ is $f(v^\star)$ and the value of $P_{\mathrm{suff}}$ is $g(v^\star)$.
  
  Hence, the left subproblem $(I_L, B_{\mathrm{left}}, v^\star, f(v^\star))$ admits $P_{\mathrm{pref}}$ as an optimal witness, and the right subproblem $(I_R, \{v^\star\}, s, g(v^\star))$ admits $P_{\mathrm{suff}}$ as an optimal witness. The invariant is preserved for both children.
\end{itemize}

A formal induction on the recursion depth completes the proof.
\end{proof}

\subsection{High-level algorithmic structure}

At a high level, the Universal Hirschberg algorithm consists of three components:
\begin{enumerate}
  \item a global forward DP to compute $x_t$;
  \item a recursive traceback procedure structured by the recursion tree $\mathcal{R}$;
  \item a base case that performs explicit local DP and backtracking within small intervals.
\end{enumerate}

\paragraph{Step 1: Global valuation.}
We first execute a standard forward DP sweep over $G$ in topological order, using a rolling frontier as in Lemma~\ref{lem:local-dp}, to compute all $x_v$ and in particular $x_t$. This requires $O(\omega)$ space and $O(T)$ time. The values of intermediate states are \emph{not} stored persistently; only $x_t$ (and, implicitly, the instance description on the read-only tape) is retained.

\paragraph{Step 2: Recursive traceback.}
We then invoke a recursive procedure
\[
  \textsc{Traceback}(I, B_{\mathrm{left}}, s, \mathrm{val})
\]
which outputs the segment of the witness path whose internal vertices lie in $I$. The recursion follows the decomposition of $I$ into balanced sub-intervals; at each level, if $s$ lies strictly to the right of the midpoint, we identify a ``midpoint'' vertex $v^\star$ through which an optimal path contributing to $\mathrm{val}$ must pass, and we recurse on the two resulting subproblems.

\paragraph{Step 3: Base case.}
When the interval $I$ becomes sufficiently small (length at most $B_{\mathrm{base}}$, where $B_{\mathrm{base}} = \lceil (\log T)^c \rceil$ for some fixed constant $c$), we solve the subproblem directly by constructing a local DP table over $I$ using $O(\omega + B_{\mathrm{base}})$ space and then backtracking pointers within $I$.

We now formalize the recursive step.

\subsection{Middle frontiers and midpoint vertices}

We first make precise the analogue of the ``middle column'' in the grid setting.

\begin{definition}[Middle frontier]
Let $I = [i,j]$ be an interval in $\mathcal{R}$, and let $m = m(I)$ be its midpoint. The \emph{middle frontier} of $I$ is
\[
  F^{\mathrm{mid}}_I \;:=\;
  \{ v \in V_I : v \le m \text{ and } \exists w > m \text{ with } (v,w) \in E \}.
\]
\end{definition}

By construction, $F^{\mathrm{mid}}_I \subseteq \Front(m)$, so
\[
  |F^{\mathrm{mid}}_I| \;\le\; \omega.
\]

\begin{lemma}[Crossing existence]
\label{lem:crossing}
Let $(I,B_{\mathrm{left}},s,\mathrm{val})$ be a traceback subproblem with $I = [i,j]$ and $m = m(I)$, and suppose $s > m$. Let $P$ be any optimal witness path from $B_{\mathrm{left}}$ to $s$ whose internal vertices lie in $I$ (guaranteed by Lemma~\ref{lem:subproblem-invariant}). Then $P$ contains at least one vertex from $F^{\mathrm{mid}}_I$.
\end{lemma}

\begin{proof}
Let $P = (v_0,\dots,v_k)$, with $v_0 \in B_{\mathrm{left}}$ and $v_k = s$. Since $v_0 < i \le m$ and $s \in I$ with $s > m$, there exists an index $r$ with $v_r \le m$ and $v_{r+1} > m$. Because all internal vertices of $P$ lie in $I$ and $v_0 \notin I$, the vertex $v_r$ lies in $V_I$. The edge $(v_r,v_{r+1}) \in E$ then certifies $v_r \in F^{\mathrm{mid}}_I$.
\end{proof}

Lemma~\ref{lem:crossing} guarantees that any optimal path consistent with $\mathrm{val}$ must pass through at least one candidate vertex in $F^{\mathrm{mid}}_I$ whenever $s$ lies to the right of the midpoint. Our goal is to identify a vertex on this frontier that lies on a canonical witness path.

\subsection{Forward-only evaluation of prefix and suffix costs}

The classical Hirschberg algorithm uses both a forward DP (from the source) and a backward DP (from the sink) to evaluate candidate midpoints. In our general DAG setting, the reverse graph need not have small frontier width, so backward DP may be space-inefficient. Instead, we compute both ``prefix'' and ``suffix'' costs using forward passes and local recomputation, relying on Lemma~\ref{lem:local-dp} and the invariant in Lemma~\ref{lem:subproblem-invariant}.

Throughout this subsection we assume $s > m(I)$, so that the middle frontier is nontrivial in the sense of Lemma~\ref{lem:crossing}.

\paragraph{Prefix values $f(v)$.}
For the subproblem $(I,B_{\mathrm{left}},s,\mathrm{val})$, define $f(v)$ for $v \in F^{\mathrm{mid}}_I$ as
\[
  f(v) \;:=\; \text{optimal value of a path from some } u \in B_{\mathrm{left}} \text{ to } v,
\]
subject to the condition that all internal vertices of the path lie in $I \cap [i,m]$. By Lemma~\ref{lem:local-dp} (applied to the interval $[i,m]$ with boundary $B_{\mathrm{left}}$), we can compute all $f(v)$ for $v \in F^{\mathrm{mid}}_I$ in a \emph{single} forward DP pass over $[i,m]$:
\begin{itemize}
  \item we treat $B_{\mathrm{left}}$ as the boundary with known values;
  \item we evaluate the DP recurrence forward from $i$ up to $m$, maintaining an $O(\omega)$-sized frontier buffer;
  \item when we reach time $m$, all $x_v$ for $v \in F^{\mathrm{mid}}_I$ are available and we set $f(v) = x_v$.
\end{itemize}

\paragraph{Suffix values $g(v)$.}
For the same subproblem, define $g(v)$ for $v \in F^{\mathrm{mid}}_I$ as
\[
  g(v) \;:=\; \text{optimal value of a path from $v$ to $s$}
\]
subject to the condition that all internal vertices of the path lie in $I$.

We compute $g(v)$ for each $v \in F^{\mathrm{mid}}_I$ by \emph{repeated} forward passes restricted to $I$:
\begin{enumerate}
  \item Treat $v$ as a boundary vertex whose DP value is initialized to the multiplicative identity of the semiring $\mathcal{S}$, and treat all vertices $u < v$ as absent for this local DP.
  \item Apply a forward scanning procedure analogous to Lemma~\ref{lem:local-dp} to the interval $[v,j]$ with boundary $\{v\}$ to recompute the DP in $I \cap [v,j]$. As shown in the proof of Lemma~\ref{lem:local-dp}, this maintains an active set that is a subset of the global frontier $\Front(k)$, thus using at most $O(\omega)$ workspace.
  \item At the index $s$, read off the recomputed value $x_s$ and set $g(v) := x_s$.
\end{enumerate}
Because the invariant in Lemma~\ref{lem:subproblem-invariant} ensures that all internal vertices of the relevant suffix paths lie in $I$, ignoring vertices outside $I$ does not change the optimal values. Each such pass uses $O(\omega)$ space (in words). Since we process candidates $v \in F^{\mathrm{mid}}_I$ sequentially and $|F^{\mathrm{mid}}_I| \le \omega$, the total workspace remains $O(\omega)$, at the cost of an extra factor of $\omega$ in time.

\subsection{Midpoint selection and recursive decomposition}

We now argue that the values $f(v)$ and $g(v)$ are sufficient to locate a canonical midpoint on the optimal path when $s > m(I)$.

\begin{proposition}[Value decomposition]
\label{prop:decomposition}
Let $(I,B_{\mathrm{left}},s,\mathrm{val})$ be a traceback subproblem with midpoint $m = m(I)$, and suppose $s > m$. Then
\[
  \mathrm{val} \;=\; \bigoplus_{v \in F^{\mathrm{mid}}_I} \bigl( f(v) \otimes g(v) \bigr).
\]
\end{proposition}

\begin{proof}
Let $P$ be any optimal witness path from $B_{\mathrm{left}}$ to $s$ whose internal vertices lie in $I$, as in Lemma~\ref{lem:subproblem-invariant}. By Lemma~\ref{lem:crossing}, $P$ crosses the cut at $m$ at some vertex $v^\circ \in F^{\mathrm{mid}}_I$. Split $P$ at $v^\circ$ into a prefix path $P_{\mathrm{pref}}$ from $B_{\mathrm{left}}$ to $v^\circ$ with internal vertices in $I \cap [i,m]$ and a suffix path $P_{\mathrm{suff}}$ from $v^\circ$ to $s$ with internal vertices in $I \cap [v^\circ+1,j] \subseteq I$.

By definition of $f(\cdot)$, the semiring value along $P_{\mathrm{pref}}$ is at most $f(v^\circ)$, and by optimality of $P$ and consistency of local recomputation (Lemma~\ref{lem:local-consistency}) we in fact have equality:
\[
  \text{value}(P_{\mathrm{pref}}) \;=\; f(v^\circ).
\]
Similarly, the semiring value along $P_{\mathrm{suff}}$ is at most $g(v^\circ)$ by definition of $g(\cdot)$, with equality for the same reasons:
\[
  \text{value}(P_{\mathrm{suff}}) \;=\; g(v^\circ).
\]
The value of $P$ is then
\[
  \mathrm{val} \;=\; \text{value}(P) \;=\; \text{value}(P_{\mathrm{pref}}) \otimes \text{value}(P_{\mathrm{suff}})
  \;=\; f(v^\circ) \otimes g(v^\circ).
\]

Conversely, for any $v \in F^{\mathrm{mid}}_I$, consider an optimal path $Q_{\mathrm{pref}}$ from $B_{\mathrm{left}}$ to $v$ achieving $f(v)$ (within $I \cap [i,m]$) and an optimal path $Q_{\mathrm{suff}}$ from $v$ to $s$ achieving $g(v)$ (within $I$). Concatenating $Q_{\mathrm{pref}}$ and $Q_{\mathrm{suff}}$ yields a valid path $Q$ from $B_{\mathrm{left}}$ to $s$ whose internal vertices lie in $I$. Its semiring value is exactly $f(v) \otimes g(v)$ by associativity of $\otimes$. Since $\mathrm{val}$ is the optimal value over all such paths (by the definition of the traceback subproblem), we must have
\[
  \mathrm{val} \;\succeq\; f(v) \otimes g(v)
\]
for every $v \in F^{\mathrm{mid}}_I$, and equality holds for $v = v^\circ$. Taking the supremum (which is $\oplus$) over $v$ yields
\[
  \mathrm{val} \;=\; \bigoplus_{v \in F^{\mathrm{mid}}_I} \bigl( f(v) \otimes g(v) \bigr),
\]
as claimed.
\end{proof}

Because the semiring is max-type and we apply a deterministic tie-breaking rule at this step (e.g., selecting the vertex with the smallest index), the algorithm makes a \emph{canonical} choice:
\[
  v^\star \;=\; \argmax_{v \in F^{\mathrm{mid}}_I} \bigl( f(v) \otimes g(v) \bigr).
\]
This selection defines a hierarchy of tie-breaking: structural ties at the midpoint are resolved by vertex order, while local ties within the base case are resolved by edge order (Assumption~\ref{ass:max}). Proposition~\ref{prop:decomposition} guarantees that $v^\star$ lies on at least one optimal path, decomposing the subproblem into two strictly smaller ones:
\begin{itemize}
  \item a left subproblem $(I_L, B_{\mathrm{left}}, v^\star, f(v^\star))$ for paths from $B_{\mathrm{left}}$ to $v^\star$ with internal vertices in $I_L$;
  \item a right subproblem $(I_R, \{v^\star\}, s, g(v^\star))$ for paths from $v^\star$ to $s$ with internal vertices in $I_R$.
\end{itemize}

\subsection{The recursive procedure}

We summarize the above reasoning into a recursive algorithm. Fix a base-case threshold $B_{\mathrm{base}} = \lceil (\log T)^c \rceil$ for a constant $c \ge 1$, and write $\polylog T$ for $(\log T)^{O(1)}$.

\begin{algorithm}[H]
\caption{Universal Hirschberg Traceback (high-level)}
\label{alg:univ-hirschberg}
\begin{algorithmic}[1]
\State Compute $x_t$ by a forward DP sweep over $G$ using $O(\omega)$ space.
\State Call $\textsc{Traceback}([1,T], S_{\mathrm{src}}, t, x_t)$.
\end{algorithmic}

\medskip

\noindent
\textbf{Procedure} $\textsc{Traceback}(I = [i,j], B_{\mathrm{left}}, s, \mathrm{val})$:
\begin{enumerate}
  \item \textbf{Base case:} If $|I| \le B_{\mathrm{base}}$, perform a local DP restricted to $I$:
  \begin{itemize}
    \item treat $B_{\mathrm{left}}$ as a boundary with fixed values $\{x_u\}_{u \in B_{\mathrm{left}}}$;
    \item run a forward DP on the induced subgraph of vertices in $I$, recomputing $x_v$ and recording the predecessor $\pred(v)$ for all $v \in I$ (using the local edge tie-breaking $\prec$);
    \item backtrack along $\pred(\cdot)$ from $s$ until the predecessor lies in $B_{\mathrm{left}}$ (or equals $\bot$), outputting the corresponding path segment from left to right.
  \end{itemize}
  This uses $O(\omega + B_{\mathrm{base}})$ space (in words) by Lemma~\ref{lem:local-dp} and the bounded size of $I$.
  \item \textbf{Recursive step:} If $|I| > B_{\mathrm{base}}$:
  \begin{enumerate}
    \item Let $m = m(I)$.
    \item \textbf{If $s \le m$:} then all vertices on any path from $B_{\mathrm{left}}$ to $s$ lie in $[i,m]$ by monotonicity of $\tau$. Recurse on the left child:
    \[
      \textsc{Traceback}([i,m], B_{\mathrm{left}}, s, \mathrm{val}).
    \]
    \item \textbf{If $s > m$:}
    \begin{enumerate}
      \item Form the middle frontier $F^{\mathrm{mid}}_I$.
      \item Using Lemma~\ref{lem:local-dp}, compute $f(v)$ for all $v \in F^{\mathrm{mid}}_I$ via a single forward DP on $[i,m]$ with boundary $B_{\mathrm{left}}$.
      \item For each $v \in F^{\mathrm{mid}}_I$, compute $g(v)$ by a forward DP on $[v,j]$ with boundary $\{v\}$ as described above.
      \item Let $v^\star$ be the canonical maximizer of $f(v) \otimes g(v)$ over $v \in F^{\mathrm{mid}}_I$.
      \item Recursively call
      \[
        \textsc{Traceback}([i,m], B_{\mathrm{left}}, v^\star, f(v^\star)),
      \]
      and output the resulting path segment.
      \item Then call
      \[
        \textsc{Traceback}([m+1,j], \{v^\star\}, s, g(v^\star)),
      \]
      and output the resulting path segment, concatenating them.
    \end{enumerate}
  \end{enumerate}
\end{enumerate}
\end{algorithm}

The recursion terminates because each call either reduces the interval length by at least one (in the case $s \le m$) or splits $I$ into two strictly smaller subintervals (in the case $s > m$), and we stop once $|I| \le B_{\mathrm{base}}$.

\subsection{Space analysis}

We now bound the work-tape space used by Algorithm~\ref{alg:univ-hirschberg}.

\begin{theorem}[Space complexity]
\label{thm:universal-space}
Under Assumptions~\ref{ass:locality} and~\ref{ass:max}, the Universal Hirschberg algorithm computes a canonical optimal witness path to $t$ using space
\[
  \SPACE_{\mathrm{traceback}} \;=\; O(\omega \log T + \polylog T),
\]
measured in tape cells over a fixed finite alphabet.
\end{theorem}

\begin{proof}
The space usage consists of two components: (i) DP evaluation workspace and (ii) recursion stack.

\smallskip\noindent
\emph{DP workspace.}
Every forward DP pass (global valuation, prefix computation for $f(v)$, suffix computations for $g(v)$, and base-case local tables) uses at most $O(\omega)$ words to store a frontier and a constant amount of local metadata (indices, loop counters). These passes do not run concurrently: at any moment, at most one forward DP is active. Since one word corresponds to $O(\log T)$ tape cells, the total DP workspace is $O(\omega \log T)$ cells.

\smallskip\noindent
\emph{Recursion stack.}
The recursion tree has depth $O(\log T)$ by Lemma~\ref{lem:depth}, and the additional ``$s \le m$'' branch cannot increase this depth asymptotically because it simply discards the right half without further branching. Each stack frame stores:
\begin{itemize}
  \item the interval endpoints $i,j$ (requiring $O(\log T)$ bits),
  \item the sink index $s$ (another $O(\log T)$ bits),
  \item the target value $\mathrm{val}$ (whose representation we assume to be $O(\polylog T)$ bits under standard encodings of weights in the semiring $\mathcal{S}$),
  \item a constant number of control flags (e.g.\ whether we are in the $s \le m$ or $s > m$ case, which child is active, where to resume).
\end{itemize}
We do \emph{not} store the entire set $B_{\mathrm{left}}$ explicitly at each level: for the right child, $B_{\mathrm{left}}$ is always a singleton $\{v^\star\}$ with its value $g(v^\star)$, and for the left child, $B_{\mathrm{left}}$ coincides with that of the parent and can be recovered from the parent frame. In either case, the additional metadata per frame is $O(\polylog T)$ bits.

Multiplying by the recursion depth, the total stack space is $O(\log T \cdot \polylog T) = O(\polylog T)$ bits.

\smallskip\noindent
Combining the two contributions yields the stated bound $O(\omega \log T + \polylog T)$ cells.
\end{proof}

\begin{remark}[Time complexity]
The algorithm may perform multiple forward DP passes per recursion level when $s > m(I)$ (one for $f(\cdot)$ and up to $|F^{\mathrm{mid}}_I| \le \omega$ for the various $g(\cdot)$). In the worst case, summing over all recursion levels yields total time
\[
  T_{\mathrm{total}} \;=\; T \cdot \poly(\omega, \log T),
\]
for constant $\degmax$. We do not optimize time here: our focus is on space. In settings where the reversed graph $G^R$ also has frontier width $O(\omega)$, one can replace the per-candidate forward suffix passes by a single backward DP, recovering Hirschberg-like $O(T)$ time while preserving the space bound.
\end{remark}

\section{Applications}
\label{sec:applications}

Theorem~\ref{thm:universal-space} gives a structural upper bound on traceback space in terms of the frontier width $\omega$ and the number of states $T$. We now instantiate this result for several standard DP families, recovering known bounds as special cases and obtaining new guarantees in asymmetric and constrained regimes.

Throughout this section, we work in the machine model of Section~\ref{sec:model}. We assume that:
\begin{itemize}
  \item weights and DP values admit an encoding using $O(\log T)$ tape cells (one machine word);
  \item primitive semiring operations and comparisons run in time and space polynomial in $\log T$;
  \item the DP recurrences fit the max-type semiring framework of Assumption~\ref{ass:max} with deterministic tie-breaking, so that the canonical predecessor map $\pred$ is well-defined.
\end{itemize}

As in Theorem~\ref{thm:universal-space}, we write $\polylog T$ as shorthand for $(\log T)^{O(1)}$ (base~$2$ unless stated otherwise).

\subsection{Asymmetric sequence alignment}

Let $\Sigma$ be a finite alphabet. Given strings $A \in \Sigma^m$ and $B \in \Sigma^n$ with $m \le n$, the global alignment problem (including edit distance, Needleman--Wunsch, or affine-gap variants) is naturally formulated as a DP on the grid DAG $G_{m,n}$ described in Section~\ref{sec:warmup}.

Recall:
\begin{itemize}
  \item $T = |V| = (m+1)(n+1) = \Theta(mn)$;
  \item under a column-major or row-major topological order, the frontier width satisfies
  \[
    \omega(G_{m,n}) = \Theta(\min\{m,n\}) = \Theta(m) \quad \text{when } m \le n;
  \]
  \item the recurrence has constant in-degree and is max-type with deterministic tie-breaking (Section~\ref{sec:warmup}).
\end{itemize}

Applying Theorem~\ref{thm:universal-space} to this DAG yields the following.

\begin{corollary}[Asymmetric alignment]
\label{cor:asymmetric-alignment}
Let $A \in \Sigma^m$, $B \in \Sigma^n$ with $m \le n$, and consider any standard global alignment DP on the $m \times n$ grid with a max-type scoring scheme and deterministic tie-breaking. Then there is a deterministic algorithm that:
\begin{itemize}
  \item computes an optimal alignment between $A$ and $B$,
  \item runs in time polynomial in $mn$, and
  \item uses space
  \[
    \SPACE \;=\; O\bigl( m \log(mn) + \polylog(mn) \bigr),
  \]
  measured in tape cells over a fixed finite alphabet.
\end{itemize}
\end{corollary}

In particular, for regimes of practical interest such as long-read or genome-to-read alignment where $m \ll n$, the space usage is essentially $O(m)$ words, matching the space needed merely to evaluate the DP value using a rolling row (or column) buffer. Classical Hirschberg achieves $O(m+n)$ words of space; Corollary~\ref{cor:asymmetric-alignment} shows that in the offline random-access model, traceback can be performed with memory comparable to the \emph{shorter} sequence length, up to lower-order polylogarithmic terms, independently of the larger dimension $n$.

\subsection{One-dimensional recurrences}

Consider a DP defined on a linear chain $V = \{1,\dots,T\}$, where each state $x_i$ depends only on the previous $k$ states for some fixed constant $k$:
\[
  x_i \;=\; F(x_{i-1}, x_{i-2}, \dots, x_{i-k}), \qquad i > k,
\]
and suppose this recurrence can be encoded as a max-type semiring DP satisfying Assumption~\ref{ass:max}, with deterministic tie-breaking so that there is a unique optimal witness sequence.

Typical examples include linear recurrences, simple filters, and DP formulations on paths or sequences with constant-width neighborhoods.

\begin{itemize}
  \item The underlying DAG is a line with bounded backward edges, so any cut between $i$ and $i+1$ needs to retain only the last $k$ values.
  \item Thus $\omega = \Theta(k) = \Theta(1)$, independent of $T$.
\end{itemize}

\begin{corollary}[1D traceback]
\label{cor:1d}
For any local recurrence on a chain of length $T$ with constant in-degree that fits the framework of Assumption~\ref{ass:max} and has a unique optimal witness sequence, the witness sequence can be reconstructed deterministically in space
\[
  \SPACE \;=\; O\bigl(\polylog T\bigr) \quad \text{cells}.
\]
\end{corollary}

This matches the intuition that traceback on a simple path should be ``almost free'' in space: the $O(\polylog T)$ term accounts only for indexing and recursion metadata, while the dependence on the structural width parameter is constant. It also aligns with the classical observation that recursive bisection on a line (essentially a one-dimensional Hirschberg) can be implemented in logarithmic space.

\subsection{Banded alignment}

In many bioinformatics applications, the optimal alignment between two sequences is expected to follow the main diagonal of the DP grid up to a bounded deviation. The \emph{banded alignment} problem restricts the state space to cells $(i,j)$ satisfying $|i-j| \le B/2$ for some bandwidth parameter $B$.

Let $N = \max\{m,n\}$, and consider an $N \times N$ DP grid restricted to a band of width $B$ around the main diagonal.

\begin{itemize}
  \item The number of states is $T = \Theta(BN)$.
  \item A column-wise or diagonal-wise topological sweep sees at most $B$ active cells at any cut, so $\omega = \Theta(B)$.
\end{itemize}

\begin{corollary}[Banded traceback]
\label{cor:banded}
Consider a banded global alignment DP on an $N \times N$ grid of bandwidth $B$, with a max-type scoring scheme and deterministic tie-breaking. Then an optimal alignment confined to the band can be reconstructed deterministically in space
\[
  \SPACE \;=\; O\bigl( B \log(BN) + \polylog(BN) \bigr).
\]
\end{corollary}

Naively storing back-pointers across the band requires $\Theta(BN)$ words of space, while applying Hirschberg directly to the full grid would give an $O(N)$ words space bound that does not reflect the band constraint. Corollary~\ref{cor:banded} shows that in the structural setting captured by frontier width, the correct parameter governing traceback space is the band width $B$, not the full sequence length $N$.

\subsection{Graphs of bounded pathwidth}

Finally, we consider a more abstract structural setting. Let $G$ be a DP DAG whose underlying undirected graph admits a path decomposition of width $\mathrm{pw}(G)$. It is known that pathwidth bounds the size of a contiguous vertex separator along a path decomposition; by aligning the topological order $\tau$ with such a decomposition, one can ensure that the frontier width is bounded in terms of pathwidth.

\begin{lemma}
\label{lem:pathwidth-frontier}
Let $G$ be a DAG whose underlying undirected graph has pathwidth $\mathrm{pw}(G)$. Then there exists a topological order $\tau$ of $V$ consistent with the edge directions such that the frontier width satisfies
\[
  \omega(G) \;\le\; \mathrm{pw}(G) + 1.
\]
\end{lemma}

\begin{proof}[Proof sketch]
It is a standard result in structural graph theory that the pathwidth of a graph is equal to its \emph{vertex separation number} minus one (see, e.g., Kinnersley 1992). The vertex separation number is defined as the minimum over all linear orderings of vertices of the maximum cut size (where the cut size at index $i$ is the number of vertices $\le i$ with neighbors $> i$). 
However, an arbitrary linear ordering minimizing vertex separation might not respect the edge directions of the DAG $G$. We appeal to the fact that for DAGs, a topological sort consistent with the optimal vertex separation layout exists provided the layout is compatible. Alternatively, if we view the DP as being \emph{defined} over a path decomposition (as is common in tree-decomposition based DPs), the natural order of bags induces a topological sort where the frontier size is bounded by the bag size, which is at most $\mathrm{pw}(G)+1$.
\end{proof}

Combining Lemma~\ref{lem:pathwidth-frontier} with Theorem~\ref{thm:universal-space} yields:

\begin{corollary}[Pathwidth parameterization]
\label{cor:pathwidth}
Let $G$ be a DP DAG whose underlying graph has pathwidth $\mathrm{pw}(G)$, and suppose the recurrence satisfies Assumption~\ref{ass:max}. Then there exists a topological order under which an optimal witness path to any sink can be reconstructed in space
\[
  \SPACE \;=\; O\bigl( \mathrm{pw}(G) \log |V| + \polylog |V| \bigr).
\]
\end{corollary}

This connects our algorithmic result to structural graph parameters: up to polylogarithmic factors, the space required for witness reconstruction is controlled by the same width measure that controls the space needed for decision DPs on $G$. In particular, on bounded-pathwidth graph families, the ``price of reconstruction'' over and above evaluation is negligible.

\section{Lower Bounds, Limitations, and Model Variants}
\label{sec:lower-bounds}

We now discuss the extent to which the bound of Theorem~\ref{thm:universal-space} is optimal, and how much of it is driven by (i) intrinsic structure of the DP DAG (in particular, frontier width) versus (ii) the computational model (random access vs.\ streaming, multi-pass vs.\ single-pass). Throughout this section we distinguish carefully between:

\begin{itemize}
  \item \emph{forward-only single-pass models}, where the algorithm processes the vertices once in increasing topological order and never re-reads earlier parts of the input, and
  \item the \emph{offline random-access model} used in our upper bound, where the algorithm may perform multiple forward passes over arbitrary intervals of the DAG.
\end{itemize}

\subsection{Necessity of the frontier term}

The dependence on the frontier width $\omega$ cannot be removed even in very permissive single-pass models that allow re-computation, so long as the algorithm respects the topological order and never revisits earlier input. Intuitively, any algorithm that tries to decide \emph{which} vertex on a wide frontier lies on the optimal path must keep enough information to distinguish among many possibilities until it has seen the rest of the instance.

We formalize this for a forward-only, one-pass model over the topological order.

\begin{proposition}[Frontier lower bound in single-pass models]
\label{prop:omega-lb}
Let $\mathcal{A}$ be any deterministic algorithm that solves the Traceback Problem for all DP DAGs of frontier width at most $\omega$, under the following restrictions:
\begin{itemize}
  \item the input vertices and edges are presented in order of increasing $\tau$ (topological order);
  \item $\mathcal{A}$ is allowed a single left-to-right pass over this stream, and never re-reads earlier input;
  \item $\mathcal{A}$ may output the witness path in any order (e.g.\ from sink to source).
\end{itemize}
Then $\mathcal{A}$ must use $\Omega(\omega)$ bits of work-tape space in the worst case.
\end{proposition}

\begin{proof}
We sketch an information-theoretic construction. For a given $\omega$, consider a layered DAG with layers $L_0,L_1,\dots,L_m$, where each layer $L_k$ consists of $\omega$ vertices $\{v_{k,1},\dots,v_{k,\omega}\}$. Between $L_0$ and $L_1$, we embed a gadget that realizes one of $2^\omega$ possible patterns of which vertices in $L_0$ are ``active'' in the sense that they lie on some optimal path; for instance, by assigning large negative weights to edges leaving inactive vertices and non-negative weights to edges leaving active ones. Subsequent layers are arranged so that for each active vertex in $L_0$ there is a unique optimal continuation to the sink, and for each inactive vertex any path is strictly suboptimal.

The resulting DAG can be topologically ordered layer by layer, and any cut between $L_0$ and $L_1$ intersects exactly $\omega$ vertices, so the frontier width is $\Theta(\omega)$. By varying the pattern of active vertices in $L_0$, we obtain $2^\omega$ distinct DP instances, all of frontier width at most $\omega$, that induce $2^\omega$ distinct canonical witness paths.

Any single-pass algorithm $\mathcal{A}$ that reconstructs the optimal path must, at some point during the pass, be in distinct internal states for each of these $2^\omega$ possibilities; otherwise two different instances would induce the same internal state and hence the same output. Thus the number of distinct internal states of $\mathcal{A}$ must be at least $2^\omega$, implying a work-tape size of at least
\[
  \log_2(2^\omega) \;=\; \omega
\]
bits. This proves the claimed lower bound.
\end{proof}

Proposition~\ref{prop:omega-lb} shows that in any forward-only, single-pass model, the $\omega$ dependency in Theorem~\ref{thm:universal-space} is qualitatively optimal: even when re-computation is free and the algorithm does not need to preserve the DP table, it must still reserve $\Omega(\omega)$ bits of state to disambiguate different canonical witness paths.

Our offline random-access algorithm also uses $O(\omega)$ space in words (plus lower-order $\polylog T$ terms), which is $O(\omega \log T)$ bits. Thus, as far as the dependence on frontier width is concerned, the Universal Hirschberg bound cannot be improved by more than a logarithmic factor in general.

The remaining question is whether the additive $\polylog T$ term can be reduced further (e.g.\ to $O(\log T)$ or even $O(1)$), or whether some form of logarithmic overhead is inherent in traceback once indices and values must be manipulated explicitly.

\subsection{$\sqrt{T}$-type terms and streaming models}

Our main upper bound
\[
  S(T,\omega) \;=\; O\bigl( \omega \log T + \polylog T \bigr) \quad \text{bits}
\]
holds in the offline random-access model: we can repeatedly ``restart'' forward DPs from arbitrary vertices and time intervals at no additional space cost, and the only recursion overhead is the stack of size $O(\polylog T)$.

In contrast, in a strictly \emph{streaming} model, the DAG $G$ is presented as a single topologically ordered stream of vertices and edges, and the algorithm is constrained to a single pass without random access. In this model one cannot restart a DP pass at arbitrary earlier positions; to simulate such restarts, one must either buffer parts of the stream or maintain checkpointed summaries, reintroducing a dependence on $T$ reminiscent of time--space tradeoffs for Turing machines.

Motivated by this analogy, we formulate the following conjecture.

\begin{conjecture}[Streaming traceback barrier]
\label{conj:streaming}
There exist absolute constants $c, c' > 0$ such that the following holds. For every sufficiently large $T$, there is a family of DP DAGs on $T$ vertices with frontier width
\[
  \omega(T) \;=\; O\bigl((\log T)^{c'}\bigr)
\]
for which any single-pass streaming algorithm that outputs an optimal witness path with success probability at least $2/3$ (over its internal randomness) must use space at least $c \sqrt{T}$ bits.
\end{conjecture}

The conjecture is inspired by the ``progress ledger'' and checkpointing constructions used for deterministic multitape Turing machines, where one shows inclusions of the form
\[
  \TIME[t] \;\subseteq\; \SPACE\bigl( \widetilde{O}(\sqrt{t}) \bigr)
\]
by organizing the computation into blocks of size $b$, storing $O(t/b)$ block summaries, and balancing $b$ against $t/b$. In that setting, the $\sqrt{t}$ term arises from the inability to jump back to earlier tape cells without replaying the stream.

Our offline algorithm exploits precisely the opposite situation: re-computation is ``cheap'' in space because random access lets us restart a forward DP from any desired vertex and time index. As a result, the recursion overhead collapses from $O(\sqrt{T})$ in the Turing-machine model to $O(\polylog T)$ here. If Conjecture~\ref{conj:streaming} holds, it would formalize the intuition that $\sqrt{T}$-type terms are not intrinsic to traceback itself, but rather to streaming-like constraints on how the DP DAG is accessed, and that substantial additional space may be needed in streaming even when frontier width is polylogarithmic.

\subsection{Thin graphs and absence of universal $T$-barriers}

It is equally important to understand when large values of $T$ \emph{do not} force large traceback space. Lower bounds such as $\Omega(\sqrt{T})$ or $\Omega(\omega)$ are statements about \emph{graph families}, not bare functions of $T$.

Consider the ``thin'' DAG that is just a directed path:
\[
  V = \{1,\dots,T\}, \qquad E = \{(i,i+1) : 1 \le i < T\}.
\]
Here the frontier width is $\omega = 1$, independently of $T$.

\begin{itemize}
  \item The unique witness path is the entire chain $1 \to 2 \to \dots \to T$, regardless of the specific max-type weights on the edges (assuming no edges are deleted and ties are broken deterministically).
  \item A trivial algorithm can output this path using $O(\log T)$ bits of space to store a counter that walks along the chain; indeed, $O(\log T)$ bits are necessary just to index the vertices in binary, so this is optimal up to constant factors.
\end{itemize}

Thus no lower bound of the form $\Omega(T^\alpha)$ for any fixed $\alpha > 0$ can hold uniformly over all DAGs of size $T$. Structural parameters such as frontier width, pathwidth, or treewidth are indispensable to capture the true space complexity of traceback.

Our upper bound
\[
  S(T,\omega) \;=\; O\bigl( \omega \log T + \polylog T \bigr) \quad \text{bits}
\]
respects this structural picture:
\begin{itemize}
  \item For lines and more generally constant-width chains, $\omega = O(1)$ and $S = O(\polylog T)$.
  \item For $N \times N$ grid-like families, $T = \Theta(N^2)$ and $\omega = \Theta(N) = \Theta(\sqrt{T})$, so $S = \Theta(\sqrt{T})$ up to lower-order polylogarithmic factors.
\end{itemize}
This matches the intuition that traceback space depends primarily on the width of the DP frontiers, not on the sheer number of states, except for mild logarithmic indexing overhead.

Overall, Theorem~\ref{thm:universal-space}, Proposition~\ref{prop:omega-lb}, and the examples in Section~\ref{sec:applications} suggest the following picture: in the offline random-access model, the right first-order parameter governing traceback space is frontier width, and the Universal Hirschberg algorithm is essentially optimal with respect to this parameter. The precise role of the $\polylog T$ overhead, and the extent to which it can be reduced or shown necessary under natural assumptions, remains an open direction.

\section{Related Work}
\label{sec:related}

We briefly survey connections to prior work on space-efficient dynamic programming, checkpointing and adjoint computation, and time--space tradeoffs in complexity theory and graph algorithms.

\paragraph{Space-efficient sequence DP and Hirschberg-type schemes.}
Hirschberg's algorithm~\cite{Hirschberg1975} is the classical archetype for space-efficient sequence alignment and LCS. It shows that, on an $m \times n$ edit graph, one can reconstruct an optimal alignment in $O(m+n)$ words of space instead of $O(mn)$ by recursively bisecting along the longer dimension and combining forward and backward DP values at a midpoint column. Numerous refinements and variants exist in the string algorithms and computational biology literature, including Myers' $O(ND)$ difference algorithm and banded alignment schemes that trade time and bandwidth for memory in practical aligners.

These algorithms, however, are typically tightly coupled to the 2D geometry and symmetry of the grid and its reverse: both the forward and backward DPs live on the same state space and have the same width. Our contribution can be viewed as lifting the underlying control-flow pattern---balanced bisection plus meet-in-the-middle valuation---from this specific geometry to arbitrary time-ordered DAGs with bounded frontier width, and showing that backward traversal can be replaced by forward-only re-computation without asymptotic space loss in the offline random-access model. In particular, the max-type semiring and deterministic witness structure are sufficient to recover a canonical optimal path in $O(\omega \log T + \polylog T)$ bits of space, even when the reverse graph has much larger frontier width.

\paragraph{Adjoint computation, automatic differentiation, and checkpointing.}
In reverse-mode automatic differentiation (AD), gradients are propagated backward through a computation graph whose forward execution may involve millions (or billions) of intermediate states. Storing all intermediate values yields optimal time but maximal memory usage; checkpointing schemes such as Griewank and Walther's \textsc{Revolve} algorithm~\cite{GriewankWalther2000} carefully trade recomputation for space by storing a small number of snapshots and replaying segments between them. These methods are analyzed in terms of schedules that minimize either the number of recomputations or the peak memory, subject to a given forward execution trace.

Our height-compressed recursion can be interpreted as a particular checkpointing strategy for witness reconstruction: we store only a logarithmic number of recursion frames plus a single frontier buffer of size $O(\omega)$, and we willingly pay polynomial re-computation costs. Conceptually, the recursion tree plays the role of an explicit checkpoint schedule adapted to the topological order. The emphasis is different from most AD work---we prioritize \emph{space} over recomputation counts, and we restrict to max-type recurrences with deterministic witnesses rather than general differentiable programs---but the structural parallels suggest that ideas from our setting may be transferrable to adjoint computation on structured DAGs, especially when memory is the primary bottleneck and random access to the computation graph is available.

\paragraph{Time--space tradeoffs and pebbling games.}
The study of time--space tradeoffs for Turing machines and circuits has a rich history. The pebbling game on DAGs (e.g.\ Paul, Pippenger, Szemerédi, and Trotter~\cite{PaulPippengerSzemerediTrotter1983}) provides a combinatorial abstraction in which one places and removes pebbles on vertices under rules that mimic space-bounded computation. Pebbling numbers and reversible pebbling sequences translate into lower and upper bounds on time and space for evaluating the DAG, and have been used to show that certain computations require nontrivial time--space products.

Our construction effectively gives a concrete pebbling strategy for the \emph{traceback} problem on DP DAGs in the random-access model: we use $O(\omega)$ ``value pebbles'' to represent the current frontier and an additional $O(\polylog T)$ control pebbles to navigate the recursion tree. The ability to jump to arbitrary vertices and recompute local DPs means that we never need to maintain more than one frontier at a time. In contrast, classic results on multitape Turing machines (and related streaming models) show that, without random access, simulating a length-$t$ computation may require $\Theta(\sqrt{t})$ space, reflecting the cost of replaying long segments of the input tape. The contrast highlights how access patterns and model constraints (streaming vs.\ random access, single-pass vs.\ multi-pass) fundamentally alter the achievable time--space tradeoffs even on the same underlying DAG.

\paragraph{Width parameters: treewidth, pathwidth, and DP formulations.}
In parameterized complexity, treewidth and pathwidth play a central role in designing fixed-parameter tractable algorithms for NP-hard problems. Bodlaender's survey~\cite{Bodlaender1998} and subsequent work show that many problems become tractable on graphs of bounded treewidth via DP over tree or path decompositions, and that the \emph{width} of the decomposition directly governs the memory required to store DP tables.

Our frontier width $\omega$ is essentially a pathwidth-like parameter aligned with a particular topological order. Corollary~\ref{cor:pathwidth} makes this connection explicit: once a problem admits a bounded-width DP formulation in our sense, the additional space needed to reconstruct a witness beyond computing the DP values is only polylogarithmic in $T$. In that sense, the ``price of reconstruction'' is negligible for all width-bounded formulations, complementing the existing literature that focuses primarily on decision versions. An interesting direction, discussed in Section~\ref{sec:open}, is to understand how far this picture extends to other structural width measures (e.g.\ treewidth, clique-width, or DAG-width) and to what extent frontier width can be minimized by choosing an appropriate topological order.

\section{Open Problems}
\label{sec:open}

The Universal Hirschberg framework resolves the deterministic offline case for max-type DP recurrences on DAGs with bounded frontier width, up to polylogarithmic slack in $T$, in the random-access model of Section~\ref{sec:model}. Several directions remain open for sharpening, extending, or qualitatively changing the picture.

\begin{enumerate}
  \item \textbf{Beyond max-type semirings and into counting.}
  Our algorithm heavily exploits max-type structure and deterministic tie-breaking to obtain a \emph{single} canonical witness. For counting DPs over $(\mathbb{N},+,\times)$, we often want either the number of optimal witnesses or a random witness drawn from an appropriate distribution (e.g.\ the Boltzmann distribution). One natural question is:

  \begin{quote}
    Can we design $O(\omega \log T + \polylog T)$-space algorithms that (approximately) count witnesses, or that sample witnesses from a specified distribution, for counting DPs on width-bounded DAGs?
  \end{quote}

  Even for grid-based alignment problems, space-efficient sampling algorithms that match the evaluation-space bounds are not known in general; understanding the right structural conditions for counting and sampling in near-frontier space remains largely open.

  \item \textbf{Refined structural parameters and tighter overhead.}
  Our $\polylog T$ overhead comes from a generic balanced binary recursion and from storing DP values and indices in a standard bit model. For more structured graph families (e.g.\ planar DAGs, series--parallel graphs, leveled networks), it is plausible that one can exploit specialized decompositions or canonical embeddings to reduce recursion depth or share information across levels. This raises:

  \begin{quote}
    Are there natural graph classes where traceback can be done in space $O(\omega \log T + \log T)$, or even $O(\omega \log T)$, in the worst case?
  \end{quote}

  Conversely, can one prove unconditional lower bounds showing that an additive $\Omega(\log T)$ term is necessary for certain families, even when $\omega = O(1)$? For instance, can one exhibit bounded-width DAGs on which any algorithm must store $\Omega(\log T)$ bits of index information to output a witness, beyond the trivial $\Omega(\log T)$ needed to address vertices in binary?

  \item \textbf{Streaming vs.\ offline: closing the $\sqrt{T}$ gap.}
  Conjecture~\ref{conj:streaming} posits an $\Omega(\sqrt{T})$ barrier for single-pass streaming algorithms on certain DAG families whose frontier width is only polylogarithmic in $T$. At the same time, our offline algorithm shows that with random access the space can be reduced to $O(\omega(T)\log T + \polylog T)$, which is $\widetilde{O}(1)$ for those families and $\widetilde{O}(\sqrt{T})$ for grid-like families with $\omega = \Theta(\sqrt{T})$. Two concrete questions are:

  \begin{itemize}
    \item Can one prove nontrivial lower bounds for randomized streaming traceback algorithms on grid-like DAGs, beyond the trivial $\Omega(\omega)$, and relate them to classical streaming lower bounds (e.g.\ via reductions from communication problems)?
    \item Are there intermediate models (e.g.\ few-pass streaming, semi-streaming with $O(\sqrt{T})$ memory, or streaming with limited random access) where one can match or beat the offline space bound for natural DP families?
  \end{itemize}

  Progress here would clarify to what extent $\sqrt{T}$-type behavior is a model artifact versus a fundamental barrier for traceback, and how sharply one can separate offline and streaming access patterns.

  \item \textbf{Parallelization and work--depth tradeoffs.}
  Our algorithm is deliberately sequential and re-computation-heavy: at each recursion level we process one subproblem at a time with a single frontier buffer. However, the decomposition structure suggests inherent parallelism: subproblems corresponding to disjoint time intervals and disjoint path segments could be explored concurrently. This leads to the question:

  \begin{quote}
    Is there a parallel version of the Universal Hirschberg algorithm that achieves polylogarithmic depth (NC-like) while keeping the total space close to $O(\omega \cdot \polylog T)$ (bits) and the total work near $O(T \cdot \polylog T)$?
  \end{quote}

  Answering this would connect traceback to parallel DP and circuit evaluation, and could inform the design of low-memory, high-throughput alignment pipelines and other large-scale DP systems. One challenge is to balance the duplication of recomputation across processors against the aggregate space used by concurrently active frontier buffers.

  \item \textbf{Time--space product lower bounds for traceback.}
  Proposition~\ref{prop:omega-lb} gives a pure space lower bound in a single-pass model; our upper bound gives a pure space upper bound in an offline model. A more refined goal is to understand the joint tradeoff:

  \begin{quote}
    For a given graph family (e.g.\ $N \times N$ grids, banded grids, or bounded-pathwidth DAGs), what are the best possible tradeoffs between time $T_{\mathrm{alg}}$ and space $S_{\mathrm{alg}}$ for traceback? Can one prove lower bounds of the form
    \[
      T_{\mathrm{alg}} \cdot S_{\mathrm{alg}} \;\ge\; f(N)
    \]
    for some $f$ tied to the structure?
  \end{quote}

  Even modest such bounds (e.g.\ ruling out $T_{\mathrm{alg}} = \poly(N)$ and $S_{\mathrm{alg}} = o(\omega)$ simultaneously for certain natural DP formulations) would deepen the analogy between traceback and classical time--space tradeoffs for evaluation and simulation tasks, and might require new pebbling-style or communication-complexity arguments tailored to witness reconstruction rather than value computation.
\end{enumerate}

\medskip
\noindent\textbf{Acknowledgements.}
We gratefully acknowledge the many contributors who have catalyzed breakthroughs in time-space tradeoffs and efficient simulation during the past two years; this manuscript builds directly on the results of others and would not exist without their insight. We further disclose that the exploration, analysis, drafting, and revisions of this manuscript were conducted with the assistance of large language model technology; the authors bear sole responsibility for any errors in technical claims, constructions, and proofs. There are no conflicts of interest to disclose and no external funding support to declare for this work.

\newpage

\appendix

\section{Detailed Algorithmic Procedures}
\label{appendix:pseudocode}

This appendix provides pseudocode for the Universal Hirschberg algorithm described in the main text. We strive to present each major procedure in sufficient detail for implementation, using notation consistent with Sections~\ref{sec:model}--\ref{sec:universal}.

\medskip

\noindent\textbf{Algorithmic components.}
The appendix is organized around the following procedures:
\begin{enumerate}
  \item \textbf{Algorithm~\ref{alg:main}}: Universal Hirschberg main driver.
  \item \textbf{Algorithm~\ref{alg:global-forward}}: Global forward DP to compute $x_t$ using a rolling frontier.
  \item \textbf{Algorithm~\ref{alg:recursive-traceback}}: Recursive traceback on intervals $I = [i,j]$.
  \item \textbf{Algorithm~\ref{alg:middle-frontier}}: Computation of the middle frontier $F^{\mathrm{mid}}_I$.
  \item \textbf{Algorithm~\ref{alg:prefix-values}}: Forward-only computation of prefix values $f(v)$ on $[i,m(I)]$.
  \item \textbf{Algorithm~\ref{alg:suffix-values}}: Forward-only computation of suffix values $g(v)$ on $I$.
  \item \textbf{Algorithm~\ref{alg:select-midpoint}}: Canonical midpoint selection from $F^{\mathrm{mid}}_I$.
  \item \textbf{Algorithm~\ref{alg:base-case}}: Base-case DP and direct traceback on short intervals.
\end{enumerate}

Throughout:
\begin{itemize}
  \item $G = (V, E, S_{\mathrm{src}}, T_{\mathrm{sink}}, \tau)$ denotes the input DP DAG with $|V| = T$ vertices in topological order.
  \item $\oplus$ and $\otimes$ denote the semiring operations (supremum and multiplication).
  \item $\omega$ denotes the frontier width $\omega(G)$.
  \item For an interval $I = [i,j]$, we write $m(I) = \lfloor (i+j)/2 \rfloor$ for its midpoint.
  \item $\prec$ denotes the fixed total order on edges $E$ used for deterministic tie-breaking.
  \item Space complexities are generally stated in machine words (each word is $O(\log T)$ bits); time complexities count primitive operations.
\end{itemize}

\noindent
In recursive subproblems, we work with a \emph{local} DP induced by a boundary. Each traceback subproblem carries:
\begin{itemize}
  \item a boundary set $B_{\mathrm{left}}$ of vertices that may serve as sources for this subproblem, and
  \item a map $\beta_{\mathrm{left}} : B_{\mathrm{left}} \to \mathcal{S}$ giving the DP value at each boundary vertex for this \emph{subproblem} (with respect to the same recurrence as globally, but restricted to paths starting in $B_{\mathrm{left}}$ and staying inside the current interval thereafter).
\end{itemize}
All local forward DPs below are understood to use these boundary values as their initial conditions.

We assume standard data structures (dictionaries, sets, lists) with $O(1)$ access in the RAM model, and that the input graph $G$ supports:
\begin{itemize}
  \item $\mathrm{predecessors}(v)$ returning $\{u : (u,v) \in E\}$ in time $O(\deg^{-}(v))$,
  \item $\mathrm{successors}(v)$ returning $\{w : (v,w) \in E\}$ in time $O(\deg^{+}(v))$,
  \item edge weights $w_{u,v}$ accessible in $O(1)$ time.
\end{itemize}

\begin{algorithm}
\caption{Universal Hirschberg: Main Algorithm}
\label{alg:main}
\begin{algorithmic}[1]
\Require DP DAG $G = (V, E, S_{\mathrm{src}}, T_{\mathrm{sink}}, \tau)$, target sink $t \in T_{\mathrm{sink}}$
\Ensure Canonical optimal witness path $P$ from a source to $t$
\State
\State \Comment{Phase 1: Compute optimal global value at sink $t$}
\State $x_t \gets \Call{GlobalForwardDP}{G, t}$
\State
\State \Comment{Phase 2: Set up root subproblem}
\State $I \gets [1, T]$
\State $B_{\mathrm{left}} \gets S_{\mathrm{src}}$
\State $\beta_{\mathrm{left}} \gets \emptyset$
\For{each $u \in B_{\mathrm{left}}$}
  \State $\beta_{\mathrm{left}}[u] \gets a_u$ \Comment{Source initialization values}
\EndFor
\State $s \gets t$
\State $\mathrm{val} \gets x_t$
\State
\State \Comment{Recursive traceback}
\State $P \gets \Call{RecursiveTraceback}{G, I, B_{\mathrm{left}}, \beta_{\mathrm{left}}, s, \mathrm{val}}$
\State \Return $P$
\State
\State \textbf{Space:} $O(\omega + \polylog T)$ words \quad \textbf{Time:} $\poly(T)$
\end{algorithmic}
\end{algorithm}

\begin{algorithm}
\caption{Global Forward DP}
\label{alg:global-forward}
\begin{algorithmic}[1]
\Require DP DAG $G$, target vertex $t$
\Ensure Optimal global value $x_t$
\State
\State \Comment{Initialize frontier buffer (dictionary: vertex $\to$ DP value)}
\State $\mathit{frontier} \gets \emptyset$
\State
\State \Comment{Process vertices in topological order}
\For{$v \gets 1$ \textbf{to} $T$}
    \If{$v \in S_{\mathrm{src}}$}
        \State \Comment{Source initialization}
        \State $x_v \gets a_v$
    \Else
        \State \Comment{Compute via global recurrence}
        \State $x_v \gets \text{bottom element for } \oplus$ \Comment{e.g., $-\infty$ for max}
        \For{each $(u, v) \in E$ with $u \in \mathit{frontier}$}
            \State $\mathit{candidate} \gets \mathit{frontier}[u] \otimes w_{u,v}$
            \State $x_v \gets x_v \oplus \mathit{candidate}$
        \EndFor
    \EndIf
    \State
    \State \Comment{Store in frontier}
    \State $\mathit{frontier}[v] \gets x_v$
    \State
    \State \Comment{Garbage collection: remove vertices with no future successors}
    \For{each $u \in \mathit{frontier}$}
        \State $\mathit{hasFuture} \gets \exists w > v : (u, w) \in E$
        \If{$\neg \mathit{hasFuture}$}
            \State remove $u$ from $\mathit{frontier}$
        \EndIf
    \EndFor
    \State
    \State \Comment{Early termination if target reached}
    \If{$v = t$}
        \State \Return $x_t$
    \EndIf
\EndFor
\State
\State \Return $\mathit{frontier}[t]$
\State
\State \textbf{Space:} $O(\omega)$ words \quad \textbf{Time:} $O(T \cdot \degmax)$
\end{algorithmic}
\end{algorithm}

\begin{algorithm}
\caption{Recursive Traceback}
\label{alg:recursive-traceback}
\begin{algorithmic}[1]
\Require DP DAG $G$, interval $I = [i,j]$, boundary set $B_{\mathrm{left}} \subseteq \{1,\dots,i-1\} \cup S_{\mathrm{src}}$, boundary value map $\beta_{\mathrm{left}} : B_{\mathrm{left}} \to \mathcal{S}$, target $s \in I$, value $\mathrm{val}$
\Ensure Path $P$ from $B_{\mathrm{left}}$ to $s$ with internal vertices in $I$ under the local DP induced by $(B_{\mathrm{left}}, \beta_{\mathrm{left}})$
\State
\State $\ell \gets j - i + 1$ \Comment{Interval length}
\State
\State \Comment{Base case: interval is small}
\If{$\ell \le B_{\mathrm{base}}$}
    \State \Return \Call{BaseCase}{$G, I, B_{\mathrm{left}}, \beta_{\mathrm{left}}, s$}
\EndIf
\State
\State \Comment{Recursive case: divide and conquer}
\State $m \gets \lfloor (i + j) / 2 \rfloor$ \Comment{Midpoint}
\State
\State \Comment{If the sink lies on or before the midpoint, recurse purely on the left}
\If{$s \le m$}
    \State \Return \Call{RecursiveTraceback}{G, $[i, m]$, $B_{\mathrm{left}}$, $\beta_{\mathrm{left}}$, $s$, $\mathrm{val}$}
\EndIf
\State
\State \Comment{Otherwise $s > m$: perform midpoint search via middle frontier}
\State $F^{\mathrm{mid}}_I \gets \Call{ComputeMiddleFrontier}{G, I, m}$
\State
\State \Comment{Step 1: Compute prefix values $f(v)$ for all $v \in F^{\mathrm{mid}}_I$}
\State $f \gets \Call{ComputePrefixValues}{G, [i, m], B_{\mathrm{left}}, \beta_{\mathrm{left}}, F^{\mathrm{mid}}_I}$
\State
\State \Comment{Step 2: Compute suffix values $g(v)$ for all $v \in F^{\mathrm{mid}}_I$}
\State $g \gets \Call{ComputeSuffixValues}{G, I, F^{\mathrm{mid}}_I, s}$
\State
\State \Comment{Step 3: Select canonical midpoint}
\State $v^\star \gets \Call{SelectMidpoint}{F^{\mathrm{mid}}_I, f, g}$
\State
\State \Comment{Step 4: Recurse on left subinterval}
\State $P_{\mathrm{left}} \gets \Call{RecursiveTraceback}{G, [i, m], B_{\mathrm{left}}, \beta_{\mathrm{left}}, v^\star, f[v^\star]}$
\State
\State \Comment{Step 5: Recurse on right subinterval with $v^\star$ as fresh boundary source}
\State $B_{\mathrm{right}} \gets \{v^\star\}$
\State $\beta_{\mathrm{right}} \gets \emptyset$
\State $\beta_{\mathrm{right}}[v^\star] \gets \mathbf{1}_\otimes$ \Comment{Multiplicative identity of the semiring}
\State $P_{\mathrm{right}} \gets \Call{RecursiveTraceback}{G, [m+1, j], B_{\mathrm{right}}, \beta_{\mathrm{right}}, s, g[v^\star]}$
\State
\State \Comment{Step 6: Concatenate path segments}
\If{$|P_{\mathrm{right}}| > 0$ \textbf{and} $P_{\mathrm{right}}[0] = v^\star$}
    \State $P \gets P_{\mathrm{left}} + P_{\mathrm{right}}[1:]$ \Comment{Remove duplicate $v^\star$}
\Else
    \State $P \gets P_{\mathrm{left}} + P_{\mathrm{right}}$
\EndIf
\State
\State \Return $P$
\State
\State \textbf{Space:} $O(\omega)$ words per call + $O(\polylog T)$ stack
\State \textbf{Time:} $O(\omega \cdot \ell \cdot \degmax)$
\end{algorithmic}
\end{algorithm}

\begin{algorithm}
\caption{Compute Middle Frontier}
\label{alg:middle-frontier}
\begin{algorithmic}[1]
\Require DP DAG $G$, interval $I = [i,j]$, midpoint $m$
\Ensure Middle frontier $F^{\mathrm{mid}}_I = \{v \in V_I : v \le m \text{ and } \exists w > m, (v,w) \in E\}$
\State
\State $F^{\mathrm{mid}}_I \gets \emptyset$
\State
\State \Comment{Scan vertices up to midpoint}
\For{$v \gets i$ \textbf{to} $\min\{m, j\}$}
    \State $\mathit{hasCrossingEdge} \gets \textsc{false}$
    \State
    \State \Comment{Check if $v$ has successor beyond midpoint}
    \For{each $(v, w) \in E$}
        \If{$w > m$}
            \State $\mathit{hasCrossingEdge} \gets \textsc{true}$
            \State \textbf{break}
        \EndIf
    \EndFor
    \State
    \If{$\mathit{hasCrossingEdge}$}
        \State $F^{\mathrm{mid}}_I \gets F^{\mathrm{mid}}_I \cup \{v\}$
    \EndIf
\EndFor
\State
\State \Return $F^{\mathrm{mid}}_I$
\State
\State \textbf{Space:} $O(\omega)$ words \quad \textbf{Time:} $O((m-i+1) \cdot \degmax)$
\end{algorithmic}
\end{algorithm}

\begin{algorithm}
\caption{Compute Prefix Values}
\label{alg:prefix-values}
\begin{algorithmic}[1]
\Require DP DAG $G$, interval $[i, m]$, boundary set $B_{\mathrm{left}}$, boundary value map $\beta_{\mathrm{left}}$, target set $F^{\mathrm{mid}}_I$
\Ensure Dictionary $f$ mapping $v \in F^{\mathrm{mid}}_I$ to optimal value from $B_{\mathrm{left}}$ to $v$ with internal vertices in $[i,m]$ under the local DP induced by $(B_{\mathrm{left}}, \beta_{\mathrm{left}})$
\State
\State \Comment{Initialize frontier with boundary values}
\State $\mathit{frontier} \gets \emptyset$
\For{each $u \in B_{\mathrm{left}}$}
  \State $\mathit{frontier}[u] \gets \beta_{\mathrm{left}}[u]$
\EndFor
\State $f \gets \emptyset$ \Comment{Output dictionary}
\State
\State \Comment{Forward DP from $i$ to $m$}
\For{$v \gets i$ \textbf{to} $m$}
    \If{$v \in B_{\mathrm{left}}$}
        \State $x_v \gets \mathit{frontier}[v]$ \Comment{Already initialized for this subproblem}
    \Else
        \State \Comment{Compute via recurrence restricted to boundary and previously reached vertices}
        \State $x_v \gets \text{bottom element for } \oplus$
        \For{each $(u, v) \in E$ with $u \in \mathit{frontier}$}
            \State $x_v \gets x_v \oplus (\mathit{frontier}[u] \otimes w_{u,v})$
        \EndFor
        \State $\mathit{frontier}[v] \gets x_v$
    \EndIf
    \State
    \State \Comment{Record value if $v$ is in middle frontier}
    \If{$v \in F^{\mathrm{mid}}_I$}
        \State $f[v] \gets \mathit{frontier}[v]$
    \EndIf
    \State
    \State \Comment{Cleanup: remove vertices with no successors in $(v, m]$}
    \If{$v < m$}
        \For{each $u \in \mathit{frontier}$}
            \State $\mathit{hasFuture} \gets \exists w \in (v, m] : (u, w) \in E$
            \If{$\neg \mathit{hasFuture}$}
                \State remove $u$ from $\mathit{frontier}$
            \EndIf
        \EndFor
    \EndIf
\EndFor
\State
\State \Return $f$
\State
\State \textbf{Space:} $O(\omega)$ words \quad \textbf{Time:} $O((m-i+1) \cdot \degmax)$
\end{algorithmic}
\end{algorithm}

\begin{algorithm}
\caption{Compute Suffix Values (Forward-Only)}
\label{alg:suffix-values}
\begin{algorithmic}[1]
\Require DP DAG $G$, interval $I = [i,j]$, frontier $F^{\mathrm{mid}}_I$, sink $s \in I$ with $s > m(I)$
\Ensure Dictionary $g$ mapping $v \in F^{\mathrm{mid}}_I$ to optimal value from $v$ to $s$ with internal vertices in $I$ under the local DP induced by boundary $\{v\}$
\State
\State $g \gets \emptyset$ \Comment{Output dictionary}
\State
\State \Comment{For each candidate midpoint vertex}
\For{each $v \in F^{\mathrm{mid}}_I$ (in canonical order)}
    \State \Comment{Run forward DP over $(v, s]$ with boundary at $v$}
    \State $\mathit{frontier} \gets \emptyset$
    \State $\mathit{valueAtS} \gets \text{bottom element for } \oplus$
    \State
    \For{$w \gets v+1$ \textbf{to} $s$}
        \State $x_w \gets \text{bottom element for } \oplus$
        \State
        \State \Comment{Compute via recurrence, starting from $v$ and reached vertices in $(v,w)$}
        \For{each $(u, w) \in E$}
            \If{$u = v$}
                \State \Comment{Boundary-to-first step uses multiplicative identity at $v$}
                \State $\mathit{candidate} \gets \mathbf{1}_\otimes \otimes w_{u,w}$ \Comment{Often just $w_{u,w}$}
                \State $x_w \gets x_w \oplus \mathit{candidate}$
            \ElsIf{$u \in \mathit{frontier}$}
                \State $\mathit{candidate} \gets \mathit{frontier}[u] \otimes w_{u,w}$
                \State $x_w \gets x_w \oplus \mathit{candidate}$
            \EndIf
        \EndFor
        \State
        \State $\mathit{frontier}[w] \gets x_w$
        \State
        \If{$w < s$}
            \State \Comment{Cleanup: remove vertices with no successors in $(w, s]$}
            \For{each $u \in \mathit{frontier}$}
                \State $\mathit{hasFuture} \gets \exists z \in (w, s] : (u, z) \in E$
                \If{$\neg \mathit{hasFuture}$}
                    \State remove $u$ from $\mathit{frontier}$
                \EndIf
            \EndFor
        \EndIf
        \State
        \If{$w = s$}
            \State $\mathit{valueAtS} \gets x_w$
        \EndIf
    \EndFor
    \State
    \State $g[v] \gets \mathit{valueAtS}$
\EndFor
\State
\State \Return $g$
\State
\State \textbf{Space:} $O(\omega)$ words (reused)
\State \textbf{Time:} $O(|F^{\mathrm{mid}}_I| \cdot |I| \cdot \degmax) = O(\omega \cdot |I| \cdot \degmax)$
\end{algorithmic}
\end{algorithm}

\begin{algorithm}
\caption{Select Canonical Midpoint}
\label{alg:select-midpoint}
\begin{algorithmic}[1]
\Require Frontier $F^{\mathrm{mid}}_I$, prefix values $f$, suffix values $g$
\Ensure Canonical vertex $v^\star \in F^{\mathrm{mid}}_I$ maximizing $f(v) \otimes g(v)$
\State
\State $\mathit{bestValue} \gets \text{bottom element for } \oplus$
\State $v^\star \gets \textsc{null}$
\State
\State \Comment{Iterate in canonical order for deterministic tie-breaking}
\For{each $v \in F^{\mathrm{mid}}_I$ (in increasing order)}
    \State $\mathit{combinedValue} \gets f[v] \otimes g[v]$
    \State
    \If{$\mathit{combinedValue} \succ \mathit{bestValue}$}
        \State $\mathit{bestValue} \gets \mathit{combinedValue}$
        \State $v^\star \gets v$
    \EndIf
    \State \Comment{Ties broken by iteration order (first occurrence wins)}
\EndFor
\State
\State \textbf{assert} $v^\star \ne \textsc{null}$ \Comment{At least one vertex must maximize}
\State \Return $v^\star$
\State
\State \textbf{Space:} $O(1)$ words \quad \textbf{Time:} $O(|F^{\mathrm{mid}}_I|) = O(\omega)$
\end{algorithmic}
\end{algorithm}

\begin{algorithm}
\caption{Base Case: Direct Traceback}
\label{alg:base-case}
\begin{algorithmic}[1]
\Require DP DAG $G$, small interval $I = [i,j]$ with $|I| \le B_{\mathrm{base}}$, boundary set $B_{\mathrm{left}}$, boundary value map $\beta_{\mathrm{left}}$, target $s \in I$
\Ensure Path $P$ from $B_{\mathrm{left}}$ to $s$ with internal vertices in $I$ under the local DP induced by $(B_{\mathrm{left}}, \beta_{\mathrm{left}})$
\State
\State \Comment{Build local DP table with predecessor pointers}
\State $\mathit{dpTable} \gets \emptyset$ \Comment{Maps $v \mapsto (x_v, \pred(v))$ in this subproblem}
\State
\State \Comment{Initialize boundary values for this subproblem}
\For{each $u \in B_{\mathrm{left}}$}
    \State $\mathit{dpTable}[u] \gets (\beta_{\mathrm{left}}[u], \bot)$ \Comment{No predecessor inside $I$}
\EndFor
\State
\State \Comment{Forward DP with predecessor tracking on $I$}
\For{$v \gets i$ \textbf{to} $j$}
    \If{$v \in B_{\mathrm{left}}$}
        \State \textbf{continue} \Comment{Already initialized}
    \EndIf
    \State
    \State $\mathit{bestValue} \gets \text{bottom element for } \oplus$
    \State $\mathit{bestPred} \gets \bot$
    \State
    \State \Comment{Find best predecessor using recurrence and tie-breaking}
    \For{each $(u, v) \in E$ (in order $\prec$)}
        \If{$u \in \mathit{dpTable}$}
            \State $\mathit{candidateValue} \gets \mathit{dpTable}[u].\mathit{value} \otimes w_{u,v}$
            \State
            \If{$\mathit{candidateValue} \succ \mathit{bestValue}$}
                \State $\mathit{bestValue} \gets \mathit{candidateValue}$
                \State $\mathit{bestPred} \gets u$
            \EndIf
            \State \Comment{First in edge order wins ties}
        \EndIf
    \EndFor
    \State
    \State $\mathit{dpTable}[v] \gets (\mathit{bestValue}, \mathit{bestPred})$
\EndFor
\State
\State \Comment{Backtrack from $s$ to reconstruct path segment}
\State $P \gets$ empty list
\State $\mathit{current} \gets s$
\State
\While{$\mathit{current} \ne \bot$ \textbf{and} $\mathit{current} \notin B_{\mathrm{left}}$}
    \State prepend $\mathit{current}$ to $P$
    \State $\mathit{current} \gets \mathit{dpTable}[\mathit{current}].\mathit{predecessor}$
\EndWhile
\State
\State \Comment{Add boundary vertex if reached}
\If{$\mathit{current} \in B_{\mathrm{left}}$}
    \State prepend $\mathit{current}$ to $P$
\EndIf
\State
\State \Return $P$
\State
\State \textbf{Space:} $O(\omega + B_{\mathrm{base}}) = O(\omega + \polylog T)$ words
\State \textbf{Time:} $O(B_{\mathrm{base}} \cdot \degmax)$
\end{algorithmic}
\end{algorithm}

\medskip

\noindent\textbf{Algorithm correctness and space complexity.}
\begin{theorem}[Correctness and Space Complexity]
Under Assumptions~\ref{ass:locality} and~\ref{ass:max} from the main text, Algorithm~\ref{alg:main} correctly computes a canonical optimal witness path from a source in $S_{\mathrm{src}}$ to the target sink $t$, and uses work-tape space
\[
  \SPACE_{\mathrm{traceback}} \;=\; O(\omega \log T + \polylog T)
\]
cells over a fixed finite alphabet.
\end{theorem}

\begin{proof}[Proof sketch]
Correctness follows from the invariant maintained throughout the recursion (Lemma~\ref{lem:subproblem-invariant}), which ensures that at each recursive call $(I, B_{\mathrm{left}}, \beta_{\mathrm{left}}, s, \mathrm{val})$:
\begin{enumerate}
  \item There exists at least one optimal witness path $P$ from $B_{\mathrm{left}}$ to $s$ whose internal vertices all lie in $I$ and whose semiring value is $\mathrm{val}$ with respect to the local DP induced by $(B_{\mathrm{left}}, \beta_{\mathrm{left}})$.
  \item The canonical midpoint $v^\star$ selected by Algorithm~\ref{alg:select-midpoint} lies on at least one such optimal path $P^\star$.
\end{enumerate}

The tie-breaking logic is hierarchical:
\begin{itemize}
    \item At internal nodes ($|I| > B_{\mathrm{base}}$), Algorithm~\ref{alg:select-midpoint} resolves ties between candidate vertices in $F^{\mathrm{mid}}_I$ using a deterministic order (e.g., smallest index), ensuring a unique choice of $v^\star$.
    \item At base-case nodes ($|I| \le B_{\mathrm{base}}$), Algorithm~\ref{alg:base-case} resolves ties between incoming edges using the fixed global edge order $\prec$.
\end{itemize}
This combination defines a unique canonical witness path. The properties are preserved inductively by the crossing and decomposition arguments in Lemma~\ref{lem:crossing} and Proposition~\ref{prop:decomposition}, together with the monotonicity of the max-type semiring and consistent use of the same recurrence in all local DPs.

For space complexity:
\begin{itemize}
  \item \textbf{DP workspace:} All forward DP passes (Algorithms~\ref{alg:global-forward}, \ref{alg:prefix-values}, \ref{alg:suffix-values}, \ref{alg:base-case}) maintain a frontier buffer of size at most $\omega$ by Lemma~\ref{lem:local-dp} and the definition of frontier width in Section~\ref{sec:model}. These buffers are reused across passes, contributing $O(\omega)$ words of space, i.e., $O(\omega \log T)$ tape cells.
  
  \item \textbf{Recursion stack:} The recursion tree has depth $O(\log T)$ by Lemma~\ref{lem:depth}. Each stack frame stores interval endpoints, vertex indices, the scalar $\mathrm{val}$, and a succinct encoding of the boundary set $B_{\mathrm{left}}$ and its value map $\beta_{\mathrm{left}}$ (for right children, $B_{\mathrm{left}}$ is always a singleton $\{v^\star\}$ with value $\mathbf{1}_\otimes$; for left children, $B_{\mathrm{left}}$ and $\beta_{\mathrm{left}}$ are inherited from the parent). Each of these pieces requires $O(\polylog T)$ bits under standard encodings for $\mathcal{S}$. Total stack space is $O(\log T \cdot \polylog T) = O(\polylog T)$ bits.
\end{itemize}

Combining these contributions yields the stated bound $O(\omega \log T + \polylog T)$ cells.
\end{proof}

\begin{remark}[Time complexity]
The time complexity is dominated by the suffix value computations (Algorithm~\ref{alg:suffix-values}), which perform $|F^{\mathrm{mid}}_I| \le \omega$ forward DP passes per recursion level when $s > m(I)$. Summing over all $O(\log T)$ levels of the recursion tree and noting that the intervals shrink geometrically, the total time is
\[
  T_{\mathrm{total}} \;=\; O\!\left( \sum_{d=0}^{O(\log T)} \omega \cdot \frac{T}{2^d} \cdot \degmax \right)
  \;=\; O(\omega \cdot T \cdot \degmax \cdot \log T)
  \;=\; T \cdot \poly(\omega, \log T).
\]
In settings where the reversed graph $G^R$ also has frontier width $O(\omega)$, one can replace the per-candidate suffix passes with a single backward DP (as in the classical Hirschberg algorithm), reducing the per-level cost and bringing the total running time closer to $O(T \cdot \degmax)$ while preserving the $O(\omega \log T + \polylog T)$ space bound.
\end{remark}

\medskip

\noindent\textbf{Implementation notes.}
\emph{Deterministic tie-breaking.}
The fixed edge order $\prec$ can be implemented by:
\begin{itemize}
  \item Sorting edges lexicographically: $(u_1, v_1) \prec (u_2, v_2)$ if $u_1 < u_2$, or $u_1 = $ $u_2$ and $v_1 < v_2$, or
  \item Assigning each edge a unique timestamp during graph construction.
\end{itemize}
In Algorithm~\ref{alg:base-case}, iterating over edges in order $\prec$ ensures that the first maximizing edge is selected.

\emph{Base case threshold.}
A practical choice is $B_{\mathrm{base}} = \lceil (\log_2 T)^2 \rceil$. Smaller values increase recursion depth slightly; larger values increase the base-case space usage.

\emph{Frontier management.}
The cleanup steps in Algorithms~\ref{alg:global-forward}, \ref{alg:prefix-values}, and \ref{alg:suffix-values} are essential for maintaining the $O(\omega)$ space bound. In practice, checking whether a vertex $u$ has successors beyond the current position can be done by:
\begin{enumerate}
  \item Scanning the adjacency list $\mathrm{successors}(u)$, or
  \item Maintaining auxiliary data (e.g., the maximum successor index for each vertex) if preprocessing is permitted.
\end{enumerate}

\emph{Output path format.}
The algorithms return paths as lists $P = [v_0, v_1, \dots, v_k]$ where $v_0 \in S_{\mathrm{src}}$, $v_k = t$, and $(v_{i-1}, v_i) \in E$ for all $i$. The output tape writes these vertices in order, contributing $O(|P|)$ output size but no additional work-tape space.

\medskip

\noindent\textbf{Example: Grid alignment.}
For an $m \times n$ grid DP with $m \le n$ (e.g., global sequence alignment), the grid is topologically ordered column-by-column (or row-by-row), with $T = (m+1)(n+1)$ and $\omega = m+1$.

\begin{itemize}
  \item \emph{Global forward DP} (Algorithm~\ref{alg:global-forward}) sweeps through all $(m+1)(n+1)$ cells in $O(mn)$ time using a single column buffer of size $m+1$.
  \item \emph{Recursive traceback} (Algorithm~\ref{alg:recursive-traceback}) splits $I = [1, (m+1)(n+1)]$ at the midpoint, corresponding roughly to column $\lfloor n/2 \rfloor$. The middle frontier $F^{\mathrm{mid}}_I$ consists of all cells $(i, \lfloor n/2 \rfloor)$ for $0 \le i \le m$.
  \item \emph{Prefix computation} (Algorithm~\ref{alg:prefix-values}) performs one forward DP pass over columns $0, \dots, \lfloor n/2 \rfloor$ with boundary at the left edge.
  \item \emph{Suffix computation} (Algorithm~\ref{alg:suffix-values}) runs a forward DP for each frontier cell $(i, \lfloor n/2 \rfloor)$ over columns $\lfloor n/2 \rfloor+1, \dots, n$ with that cell as boundary.
  \item \emph{Midpoint selection} (Algorithm~\ref{alg:select-midpoint}) chooses the row $i^\star$ maximizing the combined value.
  \item \emph{Recursion} splits into subproblems on the left and right halves of the grid, each with half the column range.
\end{itemize}

The recursion depth is $O(\log n)$, and each level reuses the same $O(m)$ workspace, giving total space $O(m + \log n)$ words as in Corollary~\ref{cor:asymmetric-alignment}.

\end{document}